\newif\iflong
\newif\ifshort

\longtrue
 
\iflong
\else
\shorttrue
\fi

\ifshort
    \documentclass[a4paper,USenglish,cleveref, autoref, thm-restate]{lipics-v2021}

    \title{SAT Backdoors: Depth Beats Size}

    \author{Jane {Open Access}}{Dummy University Computing Laboratory,
     [optional: Address], Country}{johnqpublic@dummyuni.org}{https://orcid.org/0000-0002-1825-0097}{(Optional)
       author-specific funding acknowledgements}

    \authorrunning{J. Open Access and J.\,R. Public} 

    \Copyright{Jane Open Access and Joan R. Public} 



    \category{} 

    \nolinenumbers 

    \EventEditors{John Q. Open and Joan R. Access}
    \EventNoEds{2}
    \EventLongTitle{42nd Conference on Very Important Topics (CVIT 2016)}
    \EventShortTitle{CVIT 2016}
    \EventAcronym{CVIT}
    \EventYear{2016}
    \EventDate{December 24--27, 2016}
    \EventLocation{Little Whinging, United Kingdom}
    \EventLogo{}
    \SeriesVolume{42}
    \ArticleNo{23}
\fi

\iflong
    \documentclass[letterpaper]{article}
    \usepackage[totalwidth=420pt,totalheight=625pt]{geometry}
    \usepackage[utf8]{inputenc}
    \usepackage[english]{babel}
    \usepackage{authblk,graphicx}
    \title{SAT Backdoors: Depth Beats Size}
    \author[1]{Jan Dreier}
    \author[2]{Sebastian Ordyniak}
    \author[1]{Stefan Szeider}
\affil[1]{\small Algorithms and Complexity Group, TU Wien, Vienna, Austria, \texttt{\{dreier,sz\}@ac.tuwien.ac.at}}
\affil[2]{\small Algorithms and Complexity Group, University of Leeds,
  UK, \texttt{s.ordyniak@leeds.ac.uk}}
\date{}
\fi

\iflong 
\usepackage[conf={normal,both}]{proof-at-the-end}
\fi
\ifshort
\usepackage[conf={proof at the end, restate, no link to proof}]{proof-at-the-end}
\fi

\iflong 
\newenvironment{theoremOptional}[2]
    { \begin{theoremEnd}[category=#1]{#2} }
    { \end{theoremEnd} }
\fi

\ifshort
\newenvironment{theoremOptional}[2]
    { \begin{theoremEnd}[category=#1]{#2}[$\star$] }
    { \end{theoremEnd} }
\fi

\usepackage{amsmath,amssymb,amsthm}
\usepackage{xspace}
\usepackage{hyperref}
\usepackage{cleveref}
{%
  \addtocounter{lemma}{-1}
  \endgroup
}%

\iflong
    \usepackage{enumerate}
    \newtheorem{theorem}{Theorem}
    \newtheorem{corollary}[theorem]{Corollary}
    \newtheorem{observation}[theorem]{Observation}
    \newtheorem{lemma}[theorem]{Lemma}
    \newtheorem{proposition}[theorem]{Proposition}
    \theoremstyle{definition}
    \newtheorem{definition}[theorem]{Definition}
\fi

\newtheorem{main}{Main Result}

\def\phi{\varphi}
\def\cal{\mathcal}
\def\N{\mathbf N}

\newcommand{\mpara}[1]{\smallskip\noindent\textbf{#1}}

\newcommand{\hy}{\hbox{-}\nobreak\hskip0pt}
 
\newcommand{\SB}{\{\,}%
\newcommand{\SM}{\;{|}\;}%
\newcommand{\SE}{\,\}}%

\usepackage{xcolor} 
\colorlet{MyBlue}{blue!50!black!100!}
\colorlet{MyRed}{red!50!black!100!}

\newcommand{\nat}{\mathbb{N}}

\newcommand{\SSS}{\mathsf{S}}

\newcommand{\CCC}{\mathcal{C}}

 \newcommand{\TTT}{\mathcal{T}}

\newcommand{\WWW}{\mathcal{W}}
\newcommand{\XXX}{\mathcal{X}}
\newcommand{\bad}{$\CCC$-bad\xspace}
\newcommand{\good}{$\CCC$-good\xspace}

\newcommand{\bigoh}{\mathcal{O}}

\newcommand{\Card}[1]{{|#1|}}
\newcommand{\mtext}[1]{\text{\normalfont #1}}

\newcommand{\CNF}{\mathcal{C\hspace{-0.1ex}N\hspace{-0.3ex}F}}
\newcommand{\Horn}{\mtext{\sc Horn}}
\newcommand{\QHorn}{\mtext{\sc Q-Horn}}
\newcommand{\dHorn}{\mtext{\sc dHorn}}

\newcommand{\Krom}{\mtext{\sc Krom}}
\newcommand{\EMPTY}{\mtext{\sc Null}}

\newcommand{\Bassgn}[2]{\ensuremath{#1:#2\rightarrow\{0,1\}}}

\newcommand{\satf}{F}
\newcommand{\tuple}[1]{\langle{#1}\rangle}  
\newcommand{\seq}[1]{\langle #1 \rangle}

\newcommand{\Conn}{\mathrm{Conn}}

\newcommand{\MSV}{M\"{a}hlmann et al.~\cite{MaehlmannSiebertzVigny21}\xspace}

\newcommand{\SAT}{\textsc{CnfSat}\xspace}

\newcommand{\size}[1]{{\|#1\|}}

\newcommand{\var}{\mathit{var}}
\newcommand{\srbd}{\mathrm{depth}}
\newcommand{\sbs}{\mathrm{size}}
\newcommand{\dbs}{\mathrm{del}}
\newcommand{\bdtw}{\mathrm{bdtw}}
\newcommand{\bdt}{\mathrm{leaves}}

\newcommand{\game}{\textsc{Game}(\satf, \CCC)}
\newcommand{\gameC}{\textsc{Game}(\satf, \CCC)}
\newcommand{\gameJOne}{\textsc{Game}(J_1, \CCC)}

\begin{document}
\iflong
\thispagestyle{empty}
\fi
\maketitle

\begin{abstract}
  For several decades, much effort has been put into identifying
  classes of CNF formulas whose satisfiability can be decided in
  polynomial time. Classic results are the linear-time tractability of
  Horn formulas (Aspvall, Plass, and Tarjan, 1979) and Krom (i.e.,
  2CNF) formulas (Dowling and Gallier, 1984). Backdoors, introduced by
  Williams Gomes and Selman (2003), gradually extend such a tractable
  class to all formulas of bounded distance to the class. Backdoor
  size provides a natural but
  rather crude distance measure between a formula and a tractable
  class. Backdoor depth, introduced by M\"{a}hlmann, Siebertz, and
  Vigny (2021), is a more refined distance measure, which admits the
  utilization of different backdoor variables in parallel. Bounded
  backdoor size implies bounded backdoor depth, but there are
  formulas of constant backdoor depth and arbitrarily large backdoor
  size.

  We propose FPT approximation algorithms to compute backdoor depth
  into the classes Horn and Krom. This leads to a linear-time algorithm for deciding the
  satisfiability of formulas of bounded backdoor depth into these
  classes. We base our FPT approximation algorithm on a sophisticated
  notion of obstructions, extending M\"{a}hlmann et al.'s obstruction
  trees in various ways, including the addition of separator
  obstructions. We develop the algorithm through a new game-theoretic
  framework that simplifies the reasoning about backdoors.

  Finally, we show that bounded backdoor depth captures tractable
  classes of CNF formulas not captured by any known method.
\end{abstract}

\section{Introduction}
\enlargethispage*{3mm}
Deciding the satisfiability of a propositional formula in conjunctive
normal form (\SAT) is one of the most important NP-complete problems
\cite{Cook71,Levin73}. Despite its theoretical intractability,
heuristic algorithms work surprisingly fast on real-world \SAT
instances~\cite{GaneshVardi20}. A common explanation for this
discrepancy between theoretical hardness and practical feasibility is
the presence of a certain ``hidden structure'' in realistic \SAT
instances~\cite{GomesKautzSabharwalSelman08}.
There are various approaches to capturing the vague notion of a
``hidden structure'' with a mathematical concept.  One widely studied
approach is to consider the hidden structure in terms of
decomposability. For instance, \SAT can be solved in quadratic time
for classes of CNF formulas of bounded
branchwidth~\cite{BacchusDalmaoPitassi03} or bounded treewidth~\cite{SamerSzeider10}

A complementary approach proposed by Williams et
al. \cite{WilliamsGomesSelman03} considers the hidden structure of a
\SAT instance in terms of a small number of key variables, called
\emph{backdoor variables}, that when instantiated moves the instance
into a polynomial-time solvable class. More precisely, a \emph{backdoor}\footnote{We focus on strong backdoors; we will not consider
  weak backdoors as they only apply to satisfiable formulas.}  of
\emph{size $k$} of a CNF formula $F$ into a polynomial-time solvable
class~$\CCC$ is a set $B$ of $k$ variables such that for all partial
assignments $\tau$ to $B$, the instantiated formula $F[\tau]$ belongs
to $\CCC$.
In fact, \SAT can be solved in linear time for any class of CNF
formulas that admit backdoors of \emph{bounded size} into the class
of \emph{Horn} formulas, \emph{dual Horn} formulas, and \emph{Krom}
formulas (i.e., 2CNF)\footnote{According to Schaefer's Theorem~\cite{Schaefer78}, these three classes are the largest nontrivial classes of CNF formulas defined in terms of a property of clauses, for which \SAT can be solved in polynomial time.}.
 
The size of a smallest backdoor of a CNF formula $F$ into a class $\CCC$
is a fundamental but rather simple distance measure between 
$F$ and $\CCC$.
%
M\"{a}hlmann, Siebertz, and Vigny~\cite{MaehlmannSiebertzVigny21}
proposed to consider instead the smallest \emph{depth} over all
backdoors of a formula $F$ into a class $\CCC$ as distance measure.
It is recursively defined as follows:
\begin{equation}
  \label{eq:defsrbd}
  \srbd_\CCC(F) :=
  \begin{cases}
    0   & \text{if $F\in \CCC$;}\\
    \displaystyle 1+ \min_{x\in \var(F)} \max_{\epsilon\in \{0,1\}}
    \srbd_\CCC(F[x=\epsilon]) & \text{if $F\notin \CCC$ and $F$
      is connected;}\\
    \displaystyle\max_{F'\in \Conn(F)}\srbd_\CCC(F') & \text{otherwise.}
  \end{cases}
\end{equation}
$\Conn(F)$ denotes the set of connected components of $F$; precise
definitions are given in Section~\ref{sec:prelim}.  We can certify
$\srbd_\CCC(F)\leq k$ with a \emph{component $\CCC$\hy backdoor tree}
of depth $\leq k$ which is a decision tree that reflects the choices
made in the above recursive definition.

Backdoor depth is based on the observation that if an instance $F$
decomposes into multiple connected components of $F[x=0]$ and
$F[x=1]$, then each component can be treated independently.  This way,
one is allowed to use in total an unbounded number of backdoor
variables. However, as long as the depth of the component $\CCC$\hy
backdoor tree is bounded, one can still utilize the backdoor variables
to solve the instance efficiently.  In the context of graphs, similar
ideas are used in the study of
tree-depth~\cite{NesetrilMendez06,NesetrilMendez12} and elimination
distance~\cite{bulian2016graph,fomin2021parameterized}.
Bounded backdoor size implies bounded backdoor depth, but there are
classes of formulas of unbounded backdoor size but bounded backdoor
depth.


The challenging algorithmic problem \textsc{$\CCC$\hy Backdoor Depth}
is to find for a fixed base class~$\CCC$ and a given formula $F$, a
component $\CCC$\hy backdoor tree of $F$  of depth $\leq k$.
\MSV gave an FPT-approximation algorithm for this problem, with $k$ as
the parameter) where $\CCC$ is the trivial class $\EMPTY$ for formulas
without variables. A component $\EMPTY$\hy backdoor tree must
instantiate all variables of~$F$.

\subsection*{New Results} In this paper, we give the first
positive algorithmic results for backdoor depth into nontrivial
classes. A minimization problem admits a \emph{standard
  fixed-parameter tractable approximation (FPT-approximation)}~\cite{DBLP:journals/cj/Marx08} if for an
instance of size $n$ and parameter $k$ there is an \emph{FPT-algorithm},
i.e., an algorithm running in time $f(k)n^{\bigoh(1)}$, that either
outputs a solution of size at most $g(k)$ or outputs that the instance
has no solution of size at most $k$, for some computable functions $f$
and $g$; $g(k)$ is also referred to as the performance ratio of
the algorithm.

\begin{main}[Theorem~\ref{thm:main}]
  \textnormal{\textsc{$\CCC$\hy Backdoor Depth}} admits an
  FPT-approximation if $\CCC$ is any of the Schaefer classes Horn, dual
  Horn, or Krom.
\end{main}
   Since our FPT algorithms have linear running time for
fixed~backdoor depth~$k$, we obtain the following corollary:

\begin{main}[Corollary~\ref{cor:horn-dhorn-krom}]
  \SAT can be solved in linear time for formulas of bounded backdoor
  depth into the Schaefer classes Horn, dual Horn, and Krom.
\end{main}


\newcommand{\ort}{orthogonal\xspace}
\newcommand{\sd}{strictly dominated\xspace} 

Backdoor depth is a powerful parameter that is able to capture and
exploit structure in \SAT instances that is not captured by any other
known method. We list here some well-known parameters which render \SAT
fixed-parameter tractable (the list is not complete but covers some of
the most essential parameters). For all these parameters, there exist
CNF formulas with constant backdoor depth (into \emph{Horn},
\emph{dual Horn}, and \emph{Krom}) but where the other parameter is
arbitrarily large.  
If there also exist
formulas where the converse is true, we label the respective parameter
as ``\ort'', otherwise we label it as ``\sd.''
\begin{enumerate}
\item backdoor size into Horn, dual Horn, and Krom
  \cite{NishimuraRagdeSzeider04-informal} (\sd);
\item  number of leaves of backdoor trees into Horn, dual Horn, and
  Krom \cite{SamerSzeider08b,OrdyniakSchidlerSzeider20} (\sd);
\item  backdoor depth into the class of variable-free formulas~\cite{MaehlmannSiebertzVigny21} (\sd);
\item  backdoor treewidth to Horn, dual Horn, and Krom
  \cite{GanianRamanujanSzeider17,GanianRamanujanSzeider17b} (\ort);
\item backdoor size  into  heterogeneous base classes based on
  Horn, dual Horn, and  Krom \cite{GaspersMisraOrdyniakSzeiderZivny17} (\ort);
\item backdoor size into scattered base classes based on Horn, dual
  Horn, and Krom \cite{GanianRamanujanSzeider17a} (\ort);
\item  deletion backdoor size  into the class of quadratic Horn formulas
  \cite{GaspersOrdyniakRamanujanSaurabhSzeider16} (\ort);
\item backdoor size into bounded incidence
  treewidth~\cite{GaspersSzeider13} (\ort).
\end{enumerate}
\ifshort
We will give definitions and separation proofs in
\cref{sec:comp}.
\fi
\subsection*{Approach and Techniques} 

A common approach to construct backdoors is to compute in parallel
both an upper bound and a lower bound.  The upper bounds are obtained by
constructing the backdoor itself, and lower bounds are usually
obtained in the form of so-called \emph{obstructions}.  These are
parts of an instance that are proven to be ``far away'' from the base
class. Our results and techniques build upon the pioneering work by \MSV, who
introduce \emph{obstruction trees} for backdoor depth.  A main
drawback of their approach is that it is limited to the trivial base class $\EMPTY$, where
the obstructions are rather simple because they can contain
only boundedly many variables. Our central technical contribution is overcoming
this limitation by introducing \emph{separator obstructions}.

Separator obstructions allow us to algorithmically work with
obstruction trees containing an unbounded number of variables, an apparent
requirement for dealing with nontrivial base classes different form
$\EMPTY$.  In the context of backdoor depth, it is crucial that an
existing obstruction is disjoint from all potential future
obstructions, so they can later be joined safely into a new
obstruction of increased depth.  \MSV ensure this by placing
the whole current obstruction tree into the backdoor---an approach
that only works for the most trivial base class because only there
the obstructions have a bounded number of variables.  As one considers
more and more general base classes, one needs to construct more and
more complex obstructions to prove lower bounds.  For example, as
instances of the base class no longer have bounded diameter (of the
incidence graph of the formula) or bounded clause length,
neither have the obstructions one needs to consider.  Such
obstructions become increasingly hard to separate.  Our
separator obstructions can separate obstruction trees containing an
unbounded number of variables from all potential future obstruction
trees.  We obtain backdoors of bounded depth by combining the
strengths of separator obstructions and obstruction trees.  We further
introduce a \emph{game-theoretic framework} to reason about backdoors of
bounded depth.  With this notion, we can compute winning
strategies instead of explicitly constructing backdoors, greatly simplifying
the presentation of our algorithms.
\ifshort
\begin{center}
\emph{Due to space constrains, we provide the proofs of statements
  marked with $\star$ in the appendix.}
\end{center}
\fi

\iflong
\goodbreak
\fi

\section{Preliminaries}\label{sec:prelim}

\subsection{Satisfiability}
A \emph{literal} is a propositional variable $x$ or a negated variable
$\neg x$.  A \emph{clause} is a finite set of literals that does not contain a
complementary pair $x$ and $\neg x$ of literals. A propositional formula in
conjunctive normal form, or \emph{CNF formula} for short, is a set of
clauses. We denote by $\CNF$ the class of all CNF formulas. Let $F
\in \CNF$ and $c \in F$. We denote by $\var(c)$ the set of all
variables occurring in $c$, i.e.,
$\var(c)=\SB x \SM x\in c \lor \neg x \in c\SE$ and we set
$\var(F)=\bigcup_{c \in F}\var(c)$. For a set of literals $L$, we
denote by $\overline{L}=\SB \neg l \SM l \in L\SE$, the set of
complementary literals of the literals in $L$. The \emph{size} of a CNF formula $F$ is
$\|F\|=\sum_{c\in F}\Card{c}$.

Let $\tau : X
\rightarrow \{0,1\}$ be an assignment of some set $X$ of propositional
variables. If $X=\{x\}$ and $\tau(x)=\epsilon$, we will sometimes also denote the assignment
$\tau$ by $x=\epsilon$ for brevity. We denote by $\text{true}(\tau)$ ($\text{false}(\tau)$) the
set of all literals satisfied (falsified) by $\tau$, i.e.,
$\text{true}(\tau)=\SB x \in X \SM \tau(x)=1 \SE \cup \SB \neg
x \in \overline{X} \SM \tau(x)=0\SE$ ($\text{false}(\tau)=\overline{\text{true}(\tau)}$).
We denote by
$F[\tau]$ the formula obtained from $F$ after removing all clauses
that are satisfied by $\tau$ and from the remaining clauses removing
all literals that are falsified by $\tau$, i.e., $F[\tau]=\SB
c\setminus \text{false}(\tau) \SM c \in F \land c \cap \text{true}(\tau)=\emptyset\SE$.
We say that an assignment satisfies $F$ if $F[\tau]=\emptyset$.
We say that $F$ is \emph{satisfiable} if there is some
assignment $\tau : \var(F) \rightarrow \{0,1\}$ that satisfies $F$, otherwise $F$ is
\emph{unsatisfiable}.
$\SAT$ denotes the propositional satisfiability problem, which takes
as instance a CNF formula, and asks whether the formula is
satisfiable.

The \emph{incidence graph} of a CNF formula $F$ is the bipartite graph
$G_F$ whose vertices are the variables and clauses of $F$, and where a
variable $x$ and a clause $c$ are adjacent if and only if
$x \in \var(c)$. Via incidence graphs, graph theoretic concepts
directly translate to CNF formulas. For instance, we say that $F$ is
\emph{connected} if $G_F$ is connected, and $F'$ is a \emph{connected
  component} of $F$ if $F'$ is a maximal connected subset of $F$.
$\Conn(F)$ denotes the set of connected components of $F$.
We will also consider the \emph{primal graph} of a CNF
formula $F$, which has as vertex set $\var(F)$, and has pairs of
variables $x,y\in var(F)$ adjacent if and only of $x,y\in \var(c)$ for
some $c\in F$.

\subsection{Base classes}


\newcommand{\CpnsBeAClass}{$\cal C = \cal C_{\alpha,s}$ with $\alpha \subseteq \{+,-\}$, $\alpha \neq \emptyset$, and $s \in \nat$\xspace}

Let $\alpha \subseteq \{+,-\}$ with $\alpha\neq \emptyset$, let
$F \in \CNF$ and $c \in F$. 
We say that a literal $l$ is an \emph{$\alpha$-literal} if is a positive literal and $+\in \alpha$ or it is a negative literal and $-\in\alpha$. 
We say that a variable $v$ of $F$, \emph{$\alpha$-occurs} in $c$, if $v$ or $\neg v$ is an $\alpha$-literal that is contained in $c$.
We denote by $\var_\alpha(c)$ the set of variables that $\alpha$-occur in $c$.
For $\alpha \subseteq \{+,-\}$ with $\alpha\neq \emptyset$ and $s \in \N$, let
$\CCC_{\alpha,s}$ be the class of all CNF formulas $F$ such that every clause of $F$ contains at most $s$ $\alpha$-literals.
For $\CCC\subseteq \CNF$, we say that a clause $c$ is \emph{$\CCC$-good} if $\{c\} \in \CCC$.
Otherwise, $c$ is \emph{$\CCC$-bad}. Let $\tau$ be any (partial) assignment  of the
variables of $F$. We will frequently make use of the fact that $\CCC_{\alpha,s}$ is \emph{closed under
  assignments}, i.e., if $F \in \CCC_{\alpha,s}$, then also $F[\tau]\in \CCC_{\alpha,s}$. Therefore, whenever a clause
$c \in F$ is $\CCC_{\alpha,s}$-good it will remain $\CCC_{\alpha,s}$-good in $F[\tau]$ and
conversely whenever a clause is $\CCC_{\alpha,s}$-bad in $F[\tau]$ it is also $\CCC_{\alpha,s}$-bad in $F$.

The classes $\CCC_{\alpha,s}$ capture (according to Schaefer's Dichotomy Theorem
\cite{Schaefer78}) the largest syntactic classes of CNF formulas 
for
which the satisfiability problem can be solved in polynomial time: The
class $\CCC_{\{+\},1}=\Horn$ of \emph{Horn formulas}, the class of
$\CCC_{\{-\},1}=\dHorn$ of \emph{dual-Horn formulas}, and the class
$\CCC_{\{+,-\},2}=\Krom$ of \emph{Krom} (or \emph{2CNF})
\emph{formulas}. Note also that the class $\EMPTY{}$ of formulas
containing no variables considered by \MSV is equal to
$\CCC_{\{+,-\},0}$. We follow Williams et
al. \cite{WilliamsGomesSelman03} to focus on classes that are
closed under assignments and therefore we do not consider the classes
of 0/1-valid and affine formulas.

\iflong\enlargethispage*{5mm}\fi
Note that every class $\CCC_{\alpha,s}$ (and therefore
also the classes of Krom, Horn, and dual-Horn formulas) is trivially
\emph{linear-time recognizable}, i.e., membership in the class can be tested in linear-time. We say that a class $\CCC$ of formulas
is \emph{tractable} or \emph{linear-time tractable}, if $\SAT$
restricted to formulas in $\CCC$
can be solved in polynomial-time or linear-time, respectively.
The classes  $\Horn,\dHorn,\Krom$ are linear-time
tractable~\cite{AspvallPlassTarjan79,DowlingGallier84}.

\newcommand{\CBeAClass}{$\cal C \subseteq \CNF$\xspace}

\section{Backdoor Depth}\label{sec:srbd}

A \emph{binary decision tree} is a rooted binary tree $T$.  Every
inner node $t$ of $T$ is assigned a propositional variable, denoted by
$\var(t)$, and has exactly one left and one right child, which
correspond to setting the variable to $0$ or $1$,
respectively. Moreover, every variable occurs at most once on any
root-to-leaf path of $T$. We denote by $\var(T)$ the set of all
variables assigned to any node of $T$. Finally, we associate with each
node $t$ of $T$, the truth assignment $\tau_t$ that is defined on all
the variables $\var(P)\setminus \{\var(t)\}$ occurring on the unique path~$P$ from the root
of $T$ to $t$ such that $\tau_t(v)=0$ ($\tau_t(v)=1$) if
$v \in \var(P)\setminus \{\var(t)\}$ and $P$ contains the left child
(right child) of the node $t'$ on $P$ with $\var(t')=v$.
Let $\CCC$ be a base class, $F$ be a CNF formula, and~$T$ be a decision tree
with~$\var(T)\subseteq \var(F)$. Then~$T$ is a \emph{$\CCC$\hy
  backdoor tree of $F$} if~$F[\tau_t]\in \CCC$ for every leaf~$t$
of~$T$ \cite{SamerSzeider08b}.
  
\emph{Component backdoor trees} generalize backdoor trees as
considered by Samer and Szeider \cite{SamerSzeider08b} by allowing an
additional type of nodes, \emph{component nodes}, where the current
instance is split into connected components.  More precisely, let
$\CCC$ be a base class and $\satf$ a CNF formula. A \emph{component
  $\CCC$\hy backdoor tree for $F$} 
  is a pair $(T,\phi)$, where $T$ is a
rooted tree and $\phi$ is a mapping that assigns each node $t$ a CNF
formula $\phi(t)$ such that the following conditions are
satisfied:
\begin{enumerate}
\item For the root $r$ of $T$, we have $\phi(r)=\satf$.
\item For each leaf $\ell$ of $T$, we have $\phi(\ell)\in \CCC$.
\item For each non-leaf $t$ of $T$, there are two possibilities:
  \begin{enumerate}
  \item $t$ has exactly two children $t_0$ and $t_1$, where for some
    variable $x\in \var(\phi(t))$ we have
    $\phi(t_i)=\phi(t)[x=i]$; in this
    case we call $t$ a \emph{variable node}.
  \item $\Conn(\phi(t))=\{F_1,\dots,F_k\}$ for $k\geq 2$ and $t$ has
    exactly $k$ children $t_1,\dots,t_k$ with $\phi(t_i)=F_i$; in this
    case we call $t$ a \emph{component node}.
  \end{enumerate}
\end{enumerate}
Thus, a backdoor tree as considered by Samer and Szeider
\cite{SamerSzeider08b} is just a component backdoor tree without
component nodes. The \emph{depth} of a $\CCC$-backdoor is
the largest number of variable nodes on any root-to-leaf path in the
tree. The \emph{$\CCC$\hy backdoor depth}
$\srbd_\CCC(F)$ of a CNF formula $F$ into a base class $\CCC$ is the
smallest depth over all component $\CCC$\hy backdoor trees
of~$F$. Alternatively, we can define the $\CCC$\hy
backdoor depth recursively as in \eqref{eq:defsrbd}.
For a component backdoor tree $(T,\phi)$ let $\var(T,\phi)$ denote
the set of all variables $x$ such that that some variable node $t$ of
$T$ branches on~$x$. 
The next lemma shows how to use a component $\CCC$\hy backdoor tree to
decide the satsifiability of a formula $F$.
\begin{theoremOptional}{ssrbd}{lemma}\label{lem:runtime}
  Let $\CCC\subseteq \CNF$ be (linear-time) tractable, let $F \in \CNF$, and let $(T,\phi)$ be a component $\CCC$\hy
  backdoor tree of $F$ of depth $d$. Then, we can decide the satisfiability of $F$ in time
  $(2^{d} \size{F})^{\bigoh(1)}$ ($\bigoh(2^{d} \size{F})$). 
\end{theoremOptional}
\begin{proofEnd}
  Let $m=\size{F}$.  We start by showing that
  $\sum_{\ell \in L(T)}\size{\phi(\ell)}|\leq 2^d m$, where
  $L(T)$ denotes the set of leaves of $T$, using induction on $d$ and
  $m$.  The statement holds if $d=0$ or $m\leq 1$.  We show that it
  also holds for larger $d$ and $m$.  If the root is a variable node,
  then it has two children $c_0,c_1$, and the subtree
  rooted at any of these children represents a component $\CCC$-backdoor tree for the
  CNF formula  $\phi(c_i)$ of depth $d-1$. Therefore,
  by the induction hypothesis, we obtain that
  $s_i=\sum_{\ell \in L(T_i)}\size{\phi(\ell)}\leq 2^{d-1}m$, for
  the  subtree $T_i$ rooted at $c_i$, $i\in \{0,1\}$. Consequently,
  $\sum_{\ell \in L(T)}\size{\phi(\ell)}=s_0+s_1\leq
  2 \cdot 2^{d-1}m=2^{d}m$, as required.  If, on the other
  hand, the root is a component node, then its children, say
  $c_1,\dots,c_k$, are labeled with CNF formulas  of sizes
  $m_1+\dots+m_k = m$. Therefore, for every subtree $T_i$ of $T$
  rooted at $c_i$, we have that $T_i$ is a component $\CCC$-backdoor tree
  of depth $d$ for $\phi(c_i)$, which using the induction hypothesis
  implies that $\sum_{\ell \in L(T_i)}\size{\phi(\ell)}\leq
  2^dm_i$. Hence, we obtain
  $\sum_{\ell \in L(T)}\size{\phi(\ell)}\leq 2^dm$ in total.

  To decide the satisfiability of $F$, we first decide the
  satisfiability of all the formulas associated with the leaves of
  $T$. Because, as shown above, their total size is at most $2^dm$,
  this can be achieved in time $(2^dm)^{\bigoh(1)}$ if $\SAT$
  restricted to formulas in $\CCC$ is polynomial-time solvable and in
  time $\bigoh(2^dm)$ if $\SAT$ restricted to formulas in $\CCC$ is
  linear-time solvable. Let us call a leaf true/false if it is labeled
  by a satisfiable/unsatisfiable CNF formula, respectively. We now
  propagate the truth values upwards to the root, considering a
  component node as the logical \emph{and} of its children, and a
  variable node as the logical \emph{or} of its children. $F$ is
  satisfiable if and only if the root of~$T$ is true. We can carry out
  the propagation in time linear in the number of nodes of~$T$, which
  is linear in the number of leaves of $T$, which is at most $2^d m$.
\end{proofEnd}

Let $\CCC \subseteq \CNF$ and $F\in \CNF$.  A
(strong) \emph{$\CCC$-backdoor} of $F$ is a set $B\subseteq \var(F)$
such that $F[\tau]\in \CCC$ for each $\Bassgn{\tau}{B}$.  Assume
$\CCC$ is closed under partial assignments (which is the case for many
natural base classes and the classes $\CCC_{\alpha,s}$) and $(T,\phi)$
a component $\CCC$\hy backdoor tree of $F$. Then $\var(T,\phi)$ is a
$\CCC$\hy backdoor of~$F$.

\section{Technical Overview}\label{sec:overview}

We present all our algorithms in this work within a game-theoretic framework.
This framework builds upon the following equivalent formulation of backdoor depth using splitter games.
Similar games can be used to describe treedepth and other graph classes~\cite{GroheKS2017}.

\begin{definition}
Let \CBeAClass and $\satf \in \CNF$.
We denote by $\gameC$ 
the so-called \emph{$\CCC$-backdoor depth game on $\satf$}.
The game is played between two players, the \emph{connector} and the \emph{splitter}.
The \emph{positions} of the game are CNF formulas.
At first, the connector chooses a connected component of $\satf$ to be the \emph{starting position} of the game.
The game is over once a position in the base class $\CCC$ is reached.
We call these positions the \emph{winning positions} (of the splitter).
In each \emph{round} the game progresses from a \emph{current position} $J$ to a \emph{next position} as follows.
\begin{itemize}
\item The splitter chooses a variable $v \in \var(J)$.
\item 
    The connector chooses an assignment $\tau \colon \{v\} \to \{0,1\}$
  and a connected component $J'$ of $J[\tau]$.
  The \emph{next position} is~$J'$.
\end{itemize}
In the (unusual) case that a position $J$ contains no variables anymore but $J$ is still not in $\CCC$, the splitter looses.
For a position $J$, we denote by $\tau_J$ the assignment of all
variables assigned up to position $J$.
\end{definition}
 The following observation follows easily from the definitions of the
game and the fact that the (strategy) tree obtained by playing all possible plays
of the connector against a given $d$-round winning strategy for the splitter forms a
component backdoor tree of depth $d$, and vice versa. In particular, the
splitter choosing a variable $v$ at position $J$ corresponds to a variable node and
the subsequent choice of the connector for an assignment $\tau$ of $v$ and a
component of $J[\tau]$ corresponds to a component node (and a
subsequent variable or leaf node) in a component backdoor tree.

\begin{observation}\label{lem:strategyEqualsBackdoor}
    The splitter has a strategy for the game $\gameC$ to reach within at most $d$ rounds a winning position
    if and only if $\satf$ has $\CCC$-backdoor depth at most $d$.
\end{observation} 

Using backdoor depth games, we no longer have to explicitly construct a backdoor. 
Instead, we present algorithms that play the backdoor depth game from the perspective of the splitter.
Let us start by describing these so called \emph{splitter-algorithms} and how they can be turned into an algorithm to compute backdoor depth.
The algorithms will have some auxiliary internal state that they modify with each move.
Formally, a splitter-algorithm for the $\CCC$-backdoor depth game to a base class $\CCC$ is a procedure that
\begin{itemize}
    \item gets as input a (non-winning) position $J$ of the game, together with an internal state
    \item and returns a valid move for the splitter at position $J$, together with an updated internal state.
\end{itemize}
Assume we have a game $\gameC$ and some additional input $S$.
For a given strategy of the connector, the splitter-algorithm plays the game as one would expect:
In the beginning, the internal state is initialized with $S$ (if no additional input is given, the state is initialized empty).
Whenever the splitter should make its next move, the splitter-algorithm is
queried using the current position and internal state and afterwards the
internal state is updated accordingly.

\begin{definition}
We say a splitter-algorithm \emph{implements a strategy to reach for
a game $\gameC$ and input $S$
within at most $d$ rounds a position and internal
state with some property}
if and only if 
initializing the internal state with $S$ and then
playing $\gameC$ according to the splitter-algorithm leads---no matter what strategy the connector is using---after at most
$d$ rounds to a position and internal state with said property.
\end{definition}

The following observation converts splitter-algorithms into algorithms
for bounded depth backdoors. It builds component backdoor trees by always trying out all possible next
moves of the connector.

\begin{theoremOptional}{soverview}{lemma}\label{lem:ConvertStrategyToBackdoor}
  Let $\CCC \subseteq \CNF$ and $f_\CCC \colon \N \to \N$.
  Assume there exists a splitter-algorithm that implements a strategy to
  reach for each game $\gameC$ and non-negative integer $d$ within at most $f_\CCC(d)$ rounds either:
  \begin{enumerate}[i)]
  \item a winning position, or
  \item (an internal state representing) a proof that the $\CCC$-backdoor depth of $\satf$ is at least $d$.
  \end{enumerate}
  Further assume this splitter-algorithm always takes at most $\bigoh(\|\satf\|)$ time to compute its next move.
  Then there exists an algorithm that, given $\satf$ and $d$,
  in time at most $3^{f_\CCC(d)}\bigoh(\|\satf\|)$ either:
  \begin{enumerate}[i)]
      \item returns a component $\CCC$-backdoor tree of depth at most $f_\CCC(d)$, or
  \item concludes that
    the $\CCC$-backdoor depth of $\satf$ is at least $d$.
  \end{enumerate}
\end{theoremOptional}
\begin{proofEnd}
  We compute a component $\CCC$-backdoor
  tree of depth at most $f_\CCC(d)$ by starting at the root and then
  iteratively expanding the leaves, using the splitter-algorithm to
  compute the next variable to branch over.  For each
  position we reach, we store the internal state of the
  splitter-algorithm in a look-up table, indexed by the position.
  This way, we can easily build the component $\CCC$-backdoor tree, e.g., in a depth-first or
  breadth-first way. If we encounter at any time an internal state
  representing a proof that the $\CCC$-backdoor depth of $\satf$ is at
  least $d$, we can abort. If this is not the case, then we are
  guaranteed that every leaf represents a winning position and
  therefore an instance in $\CCC$. We have therefore found a
  component $\CCC$-backdoor tree of depth at most $f_\CCC(d)$.

  Without loss of generality, we can assume that $\satf$ is connected.  We
  need to expand the root node $\satf$ of the tree $f_\CCC(d)$ times.  We
  show that expanding a node $J$ in our tree $i$ times takes time at
  most $3^{i}c\|J\|$ for some constant $c$.  To expand a node
  $J$, we run the splitter algorithm in time $c\|J\|$ to get the next
  variable and try out both assignments for this variable.  The instance splits after an assignment into some
  components $J_1,\dots,J_k$ with $\|J_1\|+\dots+\|J_k\| \le \|J\|$.  By
  induction, we can expand component $J_j$ $i-1$ times in time
  $3^{i-1} c\|J_j\|$, getting a total run time of at most
  $c\|J\| + 2 \sum_{j} 3^{i-1} c\|J_j\| \le c\|J\| +
  2\cdot3^{i-1} c\|J\| \le 3^{i} c\|J\|$.
\end{proofEnd}

For the sake of readability, we may present splitter-algorithms
as continuously running algorithms that periodically output moves (via
some output channel) and always immediately as a reply get the next
move of the connector (via some input channel).  Such an algorithm can
easily be converted into a procedure that gets as input a position and
internal state and outputs a move and a modified internal state: The
internal state encodes the whole state of the computation, (e.g., the
current state of a Turing machine together with the contents of the
tape and the position of the head).  Whenever the procedure is called,
it ``unfreezes'' this state, performs the computation until it reaches
its next move and then ``freezes'' and returns its state together with
the move.

Our main result is an approximation algorithm (\Cref{thm:main}) that either concludes
that there is no backdoor of depth $d$, or computes a
component backdoor tree of depth at most $2^{2^{\bigoh(d)}}$.
Using \Cref{lem:ConvertStrategyToBackdoor}, we see that this is equivalent to a splitter-algorithm
that plays for $2^{2^{\bigoh(d)}}$ rounds to either reach a winning position or 
a proof that the backdoor depth is larger than $d$.

Following the approach of \MSV, our proofs of high backdoor depth come in the form of so-called \emph{obstruction trees}.
These are trees in the incidence graph of a CNF formula. Their node set therefore consists of both variables and clauses.
Obstruction trees of depth $d$ describe parts of an instance
for which the splitter needs more than $d$ rounds to win the backdoor depth game.
For depth zero, we simply take a single (bad) clause that is not allowed by the base class.
Roughly speaking, an obstruction tree of depth $d > 0$ is built from two ``separated''
obstruction trees $T_1$, $T_2$ of depth $d-1$ that are connected by a
path.
Since the two obstruction trees are separated but in the same component, we know that
for any choice of the splitter (i.e., choice of a variable $v$),
there is a response of the connector (i.e., an assignment of $v$ and a component) in which either $T_1$ or $T_2$ is whole.
Then the splitter needs by induction still more than $d-1$ additional rounds to win the game.
  
\begin{definition}\label{def:obstructionTree}
    Let $\satf \in \CNF$ and \CpnsBeAClass.
    We inductively define $\CCC$-obstruction trees $T$ for $\satf$ of increasing depth.
    \begin{itemize}
        \item 
            Let $c$ be a $\CCC$-bad clause of $\satf$.
            The set $T = \{c\}$ is a \emph{$\CCC$-obstruction tree in $\satf$ of depth 0}.
        \item
            Let $T_1$ be a $\CCC$-obstruction tree of depth~$i$ in
            $\satf$. 
            Let $\beta$ be a partial assignment of the variables in $\satf$.
            Let $T_2$ be an obstruction tree of depth $i$ in $\satf[\beta]$ such that
            no variable $v \in \var(\satf[\beta])$ occurs both in a clause of $T_1$ and $T_2$.
            Let further $P$ be (a CNF formula representing) a path that connects $T_1$ and $T_2$ in~$\satf$.
            Then $T = T_1 \cup T_2 \cup \var(P) \cup P$ is
            a \emph{$\CCC$-obstruction tree in $\satf$ of depth $i+1$}.
    \end{itemize}
\end{definition}

We will prove the following central lemma in~\Cref{sec:obstructionTrees}.

\begin{theoremOptional}{lemobstructToBackdoor}{lemma}\label{lem:obstructToBackdoor}
  Let $\satf \in \CNF$ and \CpnsBeAClass.
  If there is a $\CCC$-obstruction tree of depth $d$ in $\satf$, then
  the $\CCC$-backdoor depth of $\satf$ is larger than $d$.
\end{theoremOptional}

Our splitter-algorithm will construct obstruction trees of increasing depth by a recursive procedure (\Cref{lem:main})
that we outline now.
We say a splitter-algorithm satisfies \emph{property $i$} if it
reaches in each game $\game$ within $g_\CCC(i,d)$ rounds (for some function $g_\CCC(i,d)$) either
\begin{enumerate}[1)]
    \item a winning position, or
    \item a position $J$ and a $\CCC$-obstruction tree $T$ of depth $i$ in $\satf$ such that
          no variable in $\var(J)$ occurs in a clause of $T$, or
    \item a proof that the $\CCC$-backdoor depth of $\satf$ is at least $d$.
\end{enumerate}

If we have a splitter algorithm satisfying property $d+1$ then our main result, the approximation algorithm for backdoor depth,
directly follows from \Cref{lem:obstructToBackdoor} and \Cref{lem:ConvertStrategyToBackdoor}.
Assume we have a strategy satisfying property $i-1$, let us describe how to use it to satisfy property $i$.
If at any point we reach a winning position, or a proof that the $\CCC$-backdoor depth of $\satf$ is at least $d$,
we are done. Let us assume this does not happen, so we can focus on the much more interesting case 2).

We use property $i-1$ to construct a first tree~$T_1$ of depth $i-1$, and reach a position $J_1$.
We use it again, starting at position $J_1$ to construct a second tree $T_2$ of depth $i-1$ that is completely contained in position $J_1$.
Since $T_1$ and $T_2$ are in the same component of $F$, we can find a path $P$ connecting them.
Let $\beta$ be the assignment that assigns all the variables the splitter chose until reaching position $J_1$.
Then $T_2$ is an obstruction tree not only in $J_1$ but also in $F[\beta]$.
In order to join both trees together into an obstruction of depth $i$, we have to show,
according to \Cref{def:obstructionTree} that no variable $v \in \var(\satf[\beta])$ occurs both in a clause of $T_1$ and $T_2$.
Since no variable in $\var(J_1)$ occurs in a clause of $T_1$ (property $i-1$),
and $T_2$ was built only from $J_1$, this is the case.
The trees $T_1$ and $T_2$ are ``separated'' and can be safely joined into a new obstruction tree $T$ of depth~$i$ 
(see also \Cref{fig:main} on page~\pageref{fig:main} and the proof of \Cref{lem:main} for details).

The last thing we need to ensure is that we reach a position $J$ such that 
no variable in $\var(J)$ occurs in a clause of $T$.
This then guarantees that $T$ is ``separated'' from all future obstruction trees that we may want to join it with to satisfy property $i+1$, $i+2$ and so forth.
This is the major difficulty and main technical contribution of this paper.

It is important to note here, that the exact notion of ``separation'' between obstruction trees plays a crucial role for our
approach and is one of the main differences to \MSV.
Mählmann et al.\ solve the separation problem in a ``brute-force'' manner:
If we translate their approach to the language of splitter-algorithms,
then the splitter simply selects all variables that occur in a clause of $T$.
For their base class---the class \EMPTY{} of formulas without variables---there are at most $2^{\bigoh(d)}$ variables that occur in an obstruction tree of depth $d$.
Thus, in only $2^{\bigoh(d)}$ rounds, the splitter can select all of them, fulfilling the separation property.
This completes the proof for the base class \EMPTY.

However, already for backdoor depth to \Krom, this approach cannot work
since instances in the base class have obstruction trees with arbitrarily many clauses.
Moreover, the situation becomes even more difficult for
backdoors to \Horn{}, since additionally clauses are allowed to
contain arbitrary many literals.
Mählmann et al.\ acknowledge this as a central problem and ask for an
alternative approach to the separation problem
that works for more general base classes.


\section{Separator Obstructions}\label{sec:separatorObstructions}

The main technical contribution of this work is a separation technique that
works for the base classes $\CCC=\cal C_{\alpha,s}$. The separation
technique is based on a novel form of obstruction, which we call
\emph{separator obstruction}. 
Obstruction trees are made up of paths, therefore, it is sufficient to separate
each new path $P$ that is added to an obstruction.
Note that $P$ can be arbitrarily long and every clause on $P$ can have
arbitrary many variables and therefore the splitter cannot simply select all
variables in (clauses of) $P$.
Instead, given such a path $P$ that we want to separate, we will use
separator obstructions to develop a splitter-algorithm
(\Cref{lem:separate}) that reaches in
each game $\game$ within a bounded number of rounds either
\begin{enumerate}[1)]
\item a winning position, or
\item a position $J$ such that no variable in $\var(J)$ occurs in a clause of $P$, or
\item a proof that the $\CCC$-backdoor depth of $\satf$ is at least $d$. 
\end{enumerate}

Informally, a separator obstruction is a sequence
$\seq{P_1,\dotsc,P_\ell}$ of paths that form a tree $T_\ell$ together with an
assignment $\tau$ of certain \emph{important} variables occurring in
$T_\ell$. The variables of~$\tau$ correspond to the variables
chosen by the splitter-algorithm and the assignment $\tau$ corresponds to
the assignment chosen by the connector. Each path $P_i$ adds (at least
one) \bad clause $b_i$ to the separator
obstruction, which is an important prerequisite to increase the backdoor depth by
growing the obstruction. Moreover, by choosing the important variables
and the paths carefully, we ensure that for every \emph{outside}
variable, i.e., any variable that is not an important variable assigned by $\tau$, there
is an assignment and a component (which can be chosen by the
connector) that leaves a large enough part of
the separator obstruction intact. 
Thus, if a separator obstruction is sufficiently large, the connector can play such that even after $d$ rounds a non-empty part of the separator obstruction is still intact.
This means a large separator obstruction is a proof that the backdoor depth is larger than $d$.

To illustrate the growth of a separator obstruction (and motivate its
definition) suppose that our splitter-algorithm is at position $J$ of
the game $\game$ and has already build a separator obstruction
$X=\tuple{\seq{P_1,\dotsc,P_i},\tau}$ containing \bad clauses $b_1,\dotsc,b_i$; note that $\tau$ is compatible
with $\tau_J$. If $J$ is
already a winning position, then we are done. Therefore, $J$ has to
contain a \bad{} clause. If no \bad{} clause has a path to $T_i$ in
$J$, then $J$ satisfies 2) and we are also done. Otherwise, let $b_{i+1}$
be a \bad{} clause in $J$ that is closest to $T_i$ and let $P_{i+1}$ be a
shortest path from $b_{i+1}$ to $T_i$ in $J$. Then, we extend our separator
obstruction $X$ by attaching the path $P_{i+1}$ to $T_i$ (and obtain the tree
$T_{i+1}$). Our next
order of business is to choose a bounded number of important
variables occurring on $P_{i+1}$ that we will add to $X$. Those
variables need to be chosen in such a way that no
outside variable can destroy too much of the separator
obstruction. Apart from destroying the paths of the
separator obstruction, we also need to avoid that assigning any outside
variable makes too many of the \bad clauses $b_1,\dotsc,b_{i+1}$ \good.
Therefore, a natural choice would
be to add all variables of $b_{i+1}$ to $X$, i.e., to make those
variables important. Unfortunately, this is not possible since $b_{i+1}$ can
contain arbitrarily many literals. Instead, we will only add the
variables of $b_{i+1}$ to $X$ that $\alpha$-occur in
$b_{i+1}$. By the following lemma, the number of those variables is bounded. 

\begin{theoremOptional}{sseparatorObstructions}{lemma}\label{lem:bc-arity}
  Let $\satf \in \CNF$ and \CpnsBeAClass. If $\satf$ has
  $\CCC$-backdoor depth at most some integer $d$, then every
  clause of $\satf$ contains at most $d+s$ $\alpha$-literals.
\end{theoremOptional}
\begin{proofEnd}
  As stated in the preliminaries, we can assume that every variable occurs at most once in every clause.
  Suppose that $\satf$ contains a clause $c$ containing more than $d+s$ $\alpha$-literals.
  If the splitter chooses a variable from $c$, the connector will assign it to zero if it occurs positively in $c$
  and to one otherwise.
  Thus, the connector can play such that after $d$ rounds, $c$ still has more than $s$ $\alpha$-literals
  and therefore still is \bad.
  By \Cref{lem:strategyEqualsBackdoor}, $\satf$ has backdoor depth larger than $d$.
\end{proofEnd}

While
this still allows for outside variables to occur in many of the \bad
clauses $b_1,\dotsc,b_{i+1}$, it already ensures that no
outside variable can $\alpha$-occur in any of these clauses. This is
very helpful for our purposes, because in the case that
$|\alpha|=1$ (i.e., the only case where $\alpha$-occurs means something different then just occurs),
it provides us with an assignment of any such outside variable
that can be played by the connector without making the \bad{} clauses in which it
occurs \good{}. For instance, if $\alpha=\{+\}$, then any outside
variable $v$ can only occur negatively in a \bad clause and moreover
setting $v$ to $0$ ensures that the \bad clauses remain \bad.

The next thing that we need to ensure is that any outside variable can
not destroy too many paths.
Note that by choosing a \emph{shortest} path $P_{i+1}$, we
have already ensured that no variable occurs on more than
two clauses of $P_{i+1}$ (such a variable would be a shortcut, meaning $P_{i+1}$ was not a shortest path). Moreover, because $P_{i+1}$ is a
shortest path from $b_{i+1}$ to $T_i$, we know that every variable that
occurs on $T_i$ and on $P_{i+1}$ must occur in the clause $c$ in
$P_{i+1}$ that is closest to $T_i$ but not in $T_i$ itself. Similarly, to how we dealt
with the \bad clauses, we will now add all variables that
$\alpha$-occur in $c$ to $X$.
This ensures that no
outside variable can $\alpha$-occur in both $T_i$ and $P_{i+1}$ 
, which
(by induction over $i$) implies that every outside variable
$\alpha$-occurs in at most two clauses (either from $T_i$ or from
$P_{i+1}$) and therefore provides us with an assignment for the
outside variables that removes at most two clauses from $X$. However, since removing any
single clause can be arbitrarily bad if the clause has a high degree in
the separator obstruction, we further need to ensure that all clauses
of the separator obstruction in which outside variables $\alpha$-occur
have small degree.
We achieve this by adding the variables $\alpha$-occurring in any clause
as soon as its degree (in the separator obstruction) becomes larger than two, which happens whenever the endpoint
of $P_{i+1}$ in $T_i$ is a clause. Finally, if the endpoint of
$P_{i+1}$ in $T_i$ is a variable, we also add this variable to the
separator obstruction to ensure that no variable has degree larger
than three in $T_{i+1}$. This leads us to the following
definition of separator obstructions (see also
Figure~\ref{fig:sep-obs} for an illustration).

\begin{figure}
  \begin{center}
      \includegraphics{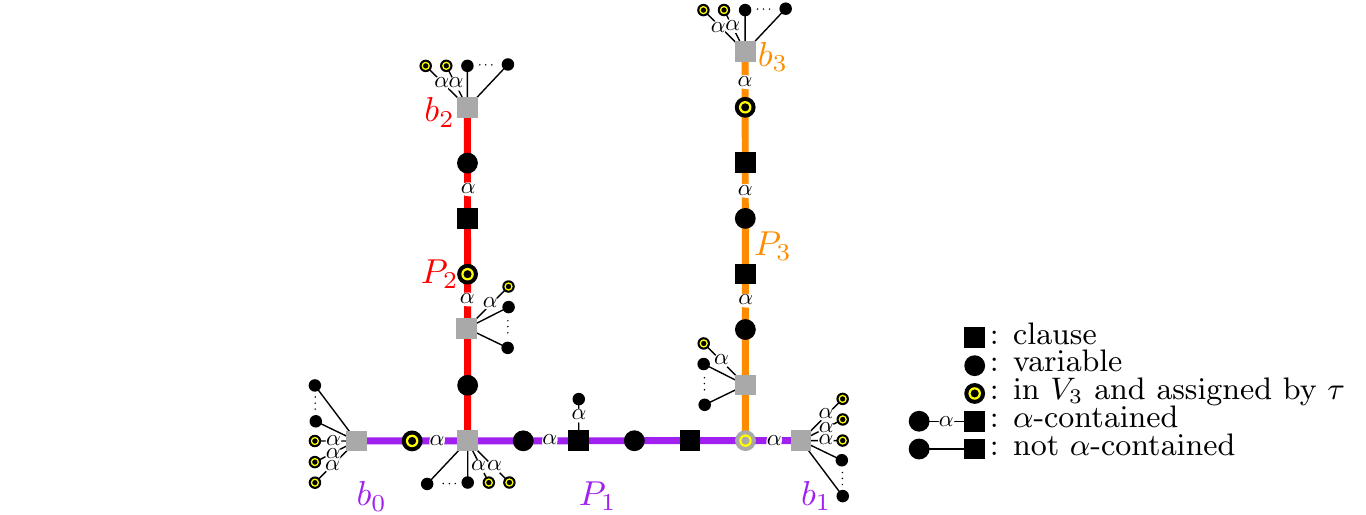}
    \end{center}
  \caption{A separator obstruction containing three paths $P_1$,
    $P_2$, and $P_3$. 
    The figure shows vertices and edges of the incidence graph.
    Only the colorful edges are part of separator obstruction's tree.
Gray variables and clauses are mentioned under the names $b_i$, $a$, and $c$ in \Cref{def:separatorObstruction}.}
  \label{fig:sep-obs}
\end{figure}

\begin{definition}\label{def:separatorObstruction}
Let $F \in \CNF$ and \CpnsBeAClass. A \emph{$\CCC$-separator obstruction} for $F$ is a tuple
$X=\tuple{\seq{P_1,\dotsc,P_\ell},\tau}$ (where $P_1,\dotsc,P_\ell$
are paths in $F$ and $\tau$ is an assignment of 
variables of $F$) satisfying the following recursive definition.
\begin{itemize}
\item $P_1$ is a shortest path between two \bad{} clauses $b_0$ and
  $b_1$ in $F$. Let $B_1=\{b_0,b_1\}$, let $V_1$ be the set of all
  variables that $\alpha$-occur in any clause in $B_1$, let
  $\tau_1 : V_1 \rightarrow \{0,1\}$ be any assignment of the
  variables in $V_1$, and let $T_1=P_1$.
\item For every $i$ with $1 < i \leq \ell$, let $b_i$ be a \bad{}
  clause in $\satf[\tau_{i-1}]$ of minimal distance to
  $T_{i-1}$ in $F[\tau_{i-1}]$. Then, $P_i$ is a shortest path (of possibly length zero) in
  $\satf[\tau_{i-1}]$ from $T_{i-1}$ to $b_i$ and $T_i=T_{i-1}\cup
  P_i$. Moreover, let $a$ be the variable or clause that is both in
  $T_{i-1}$ and $P_i$, then we define $B_i$ and $V_i$ by initially setting $B_i=B_{i-1}\cup \{b_i\}$ and $V_i=V_{i-1}
  \cup \var_\alpha(b_i)$ and distinguishing the following cases:
  \begin{itemize}
  \item If $a$ is a variable, then let $c$ be the clause on $P_i$
    incident with $a$ (note that it is possible that $c=b_i$). Then,
    we add $c$ to $B_i$ and we add $\{a\}\cup \var_\alpha(c)$ to $V_i$.
  \item If $a$ is a clause, then either $a=b_i$ or $a\neq b_i$ and
    there is a clause $c$ that is closest to $a$ on $P_i$ (note that
    it is possible that $c=b_i$). In the
    former case we leave $B_i$ and $V_i$ unchanged and in the latter case, we add $a$ and
    $c$ to $B_i$ and we add $\var_\alpha(a)\cup \var_\alpha(c)$ to $V_i$.
  \end{itemize}
  $\tau_i : V_i \rightarrow \{0,1\}$ is any assignment of the
  variables in $V_i$ that is compatible with $\tau_{i-1}$.
\item Finally, we set $\tau=\tau_\ell$.
\end{itemize}
We define the \emph{size} of $X$ to be the number of leaves in
$T=T_\ell$, i.e., $\ell+1$.
\end{definition}
 
We start by observing some simple but important
properties of separator obstructions.
\begin{theoremOptional}{sseparatorObstructions}{lemma}\label{lem:sep-obs-prop}
  Let $\satf \in \CNF$, \CpnsBeAClass, and let $X=\tuple{\seq{P_1,\dotsc,P_\ell},\tau}$ be a
  $\CCC$-separator obstruction in $\satf$, then for every
  $i \in [\ell]$:
  \begin{enumerate}[(C1)]
  \item $T_i$ is a tree. 
  \item Every variable $v \not\in V_i$ occurs in at most two
    clauses of $P_j$ for every $j$ with $1 \leq j \leq i$ and moreover those clauses are
    consecutive in $P_j$. 
  \item Every variable $v \not\in V_i$ $\alpha$-occurs in at most two
    clauses of $T_i$ and moreover those clauses are
    consecutively contained in one path of $T_i$. 
  \item Every variable $v \in V_i\setminus V_{i-1}$ $\alpha$-occurs
    in most four clauses of $T_i$.
  \item If a variable $v\not\in V_i$ $\alpha$-occurs in a clause
      $c$ of $T_i$, then $c$ has degree at most two in~$T_i$.
  \item Every variable of $\satf$ has degree at most
    three in $T$.
  \item If every clause of $\satf$ contains at most $x$ $\alpha$-literals, then $|V_i\setminus V_{i-1}|\leq 2s+x+1$.
\end{enumerate}
\end{theoremOptional}
\begin{proofEnd}
    We show (C1) by induction on $i$. (C1) clearly holds for $i=1$.
    For $i>1$, note that $T_i$ is obtained from $T_{i-1}$ by
    adding the path $P_i$ that intersects $T_{i-1}$ in at most one variable or clause. Since $T_{i-1}$ is a tree so is $T_i$.

    (C2) follows from the fact that every $P_j$ is a shortest path in $\satf[\tau_{j-1}]$ and
    because $v\notin V_i$ it holds that $v$ is a variable of $\satf[\tau_{j-1}]$ for every $j\leq i$.
    If $v$ occurred in two non-consecutive clauses of $P_j$, then $v$
    would be a ``shortcut'' and $P_j$ would not have been a shortest
    path in $\satf[\tau_{j-1}]$,  contradicting our choice of $P_j$.

    We establish (C3) by induction on $i$. 
    For $i=1$, this follows immediately from (C2) because $T_1=P_1$
    and moreover the observation that if a variable does not occur
    in a clause then it also does not $\alpha$-occur in a clause.
    Now suppose that the
    claim holds for $i-1$. Then, $v$ $\alpha$-occurs in at most two
    consecutive clauses of some path $P^v$ of $T_{i-1}$.
    Moreover, because of (C1), $v$ $\alpha$-occurs (or
    occurs ) in at most two consecutive clauses of $P_i$. We claim
    that only one of these conditions can be true, i.e., either $v$ occurs in $P_i$ but not in $T_{i-1}$ or $v$ occurs in
    $T_{i-1}$ but not in $P_i$, which shows (C3). So suppose for a
    contradiction that this is not the case and there is a clause
    $c_{i-1}$ in $T_{i-1}$ and a clause $c_i$ of $P_i$ in which $v$ $\alpha$-occurs.
    If $c_{i-1}=c_i$, then $\{c_i\}=T_{i-1}\cap P_i$ and
    therefore $c_i \in B_i$, which contradicts our assumption that $v
    \notin V_i$ (because in this case $\var_\alpha(c_i)\subseteq
    V_i$). Therefore, $c_{i-1}\neq c_i$. Moreover, because $P_i$ is a
    shortest path from $b_i$ to $T_{i-1}$ in $\satf[\tau_{i-1}]$ and
    $v \notin V_i$, it follows that $c_i$ must be the clause on $P_i$ that is
    closest to $T_{i-1}$. But then, $c_i \in B_i$, which again contradicts
    our assumption that $v \notin V_i$.

    Towards showing (C4), first note that because of
    (C3), and the fact that $v \notin V_{i-1}$, we obtain that $v$ can
    $\alpha$-occur in at most two clauses of
    $T_{i-1}$. Moreover, because $v \notin V_{i-1}$ and $P_i$ is a
    shortest path in $\satf[\tau_{i-1}]$, it follows that $v$ can
    ($\alpha$-)occur in at most two (consecutive) clauses of~$P_i$. Therefore,
    $v$ $\alpha$-occurs in at most four clauses of $T_i$.

    Towards showing (C5), first observe that if $c$ is a clause
    with degree larger than $2$ in $T_i$, then $c \in B_i$. This is
    because for $c$ to have degree larger than $2$, it must be
    contained in more than one path of $T_i$, i.e., there must be an
    index $j\leq i$ such that $c$ is contained in both $T_{j-1}$ and
    $P_j$. But then, $c \in B_j\subseteq B_i$. Now suppose for a
    contradiction that there is a clause $c$ with degree larger than
    two in which a variable $v \notin V_i$ $\alpha$-occurs. Then, $c \in
    B_i$ and because $v \in \var_\alpha(c)$, we obtain that $v \in
    V_i$, a contradiction.

    Towards showing (C6), let $v$ be any variable of $\satf$. If
    $v$ occurs in at most one path of $T$, then $v$ has degree at most
    two. Moreover, if not then let $i$ be the smallest number such
    that $v$ contained in two paths of $T_i$. Then $v$ has degree at
    most three in $T_i$ and is the
    endpoint of the path $P_i$ in $T_{i-1}$ and therefore $v$ is added
    to $V_i$. However, this implies
    that $v$ will not appear on any path $P_j$ for $j>i$ (because any
    such path $P_j$ is a path in $\satf[\tau_i]$, which does no longer
    contain $v$) and therefore
    the degree of $v$ in $T$ will be at most three.

    We finish by showing (C7).
    We say that a path $P$ of $\satf$ (i.e., a path of $G_\satf$) is
    \emph{\good} if so are all clauses occurring as inner vertices on $P$.
    Note that the paths $P_i$ for any $i>1$ in the above definition are
    necessarily \good paths due to the definition of $b_i$.
    Because of the definition of
    $\CCC$-separator obstructions, it holds that $V_i\setminus
    V_{i-1}$ is either equal to $\var_\alpha(b_i)\cup \{a\} \cup
    \var_\alpha(c)$ or equal to $\var_\alpha(b_i)\cup \var_\alpha(c') \cup
    \var_\alpha(c)$ for some \bad clause $b_i$, variable $a$, and
    \good clauses $c$ and $c'$; note that we can assume that $c$ and
    $c'$ are \good{} since for every $i>2$, $P_i$ is a \good path and
    therefore all clauses on $P_i$ apart from $b_i$ are \good. Since
    \good clauses contain at most $s$ $\alpha$-literals and by
    assumption every other clause contains at most $x$
    $\alpha$-literals, we obtain that $|V_i\setminus V_{i-1}|\leq x+2s+1$.
\end{proofEnd}


%

\newcommand{\sepobsiz}{(8^d(14^2+2d))^{2^{d}}}

We are now ready to show our main result of this subsection, namely,
that also separator obstructions can be used to obtain a lower bound on
the backdoor depth of a CNF formula.
\begin{lemma}\label{lem:sep-obs-rbd}
  Let \CpnsBeAClass and $\satf \in \CNF$.
  If $\satf$ has a $\CCC$-separator obstruction of size at least
  $\ell=\sepobsiz$, then $\satf$ has
  $\CCC$-backdoor depth at least~$d$. 
\end{lemma}
\begin{proof}
  Let $X=\tuple{\seq{P_1,\dotsc,P_\ell},\tau}$ be a $\CCC$-separator
  obstruction for $\satf$ of size at least $\ell$
  with $V_i,B_i,T_i,T$ be as in \Cref{def:separatorObstruction}.
  Let $J$ be a position in the game $\game$. We say that a subtree $T'$ of
  $T=T_\ell$ is \emph{contained} in $J$ if every variable and clause of $T'$
  occurs in $J$. Let $T'$ be a subtree of $T$ that is contained in
  $J$. Let $P_j$ be a path of $X$.
  We say that $P_j$ is \emph{active} in $T'$ if either
  $V(P_j)=\{b_j\}$ and $T'$ contains $b_j$ or $T'$ contains a vertex
  in $V(P_j)\setminus V(T_{j-1})$.
  Moreover, we say that $P_j$ is \emph{intact} in $T'$ at position $J$ if
  $V(P_j)\subseteq V(T')$ and $b_j$ is a \bad clause in $J$.
  Otherwise, we say that $P_j$ is \emph{broken} in $T'$ at position $J$. 
  
  We show by induction on the number of rounds that there is a
  strategy $\SSS$ for the connector such that the following holds for
  every position $J$ reached after $i$ rounds in the game $\game$
  against $\SSS$: At position $J$, there is a subtree $T'$ of $T$
  contained in $J$ that contains at least
  $\ell_i=(\ell^{(1/2)^{i}}/8^i)$ intact paths and at most $z_i=2i$
  broken paths of~$X$. This then shows the
  statement of the lemma because $\ell_d=\ell^{1/2^{d}}/8^d=14^2+2d \ge 1$ and
  therefore any position $J$ reached after $d$ rounds in the game $\game$
  contains at least one clause that is \bad in $J$. 

  The claim clearly holds for $i=0$ since $\ell_0=\ell$ and $z_0=0$ and the connector
  can choose the component of $\satf$ containing~$T$.
  Assume now that $i>0$ and let $J$ be the position reached after
  $i-1$ rounds. By the induction hypothesis, at position $J$ there is
  a subtree $T'$ of $T$ contained in $J$
  containing at least $\ell_{i-1}=\ell^{(1/2)^{i-1}}/8^{i-1}$ intact paths
  and at most $z_{i-1}=2(i-1)$ broken paths of $X$. Suppose that the
  splitter chooses variable $v$ as its next move. Moreover, let $o$ be
  the smallest integer such that $v \in V_o$; if $v \notin V_\ell$ we
  set $o=\ell+1$. Note that $v \notin V_{j}$ for every $i < o$. Let $I$ be the set of
  all paths $P_j$ of $X$ that are intact in $T'$ at position $J$ and let $I_{<o}$
  ($I_{>o}$) be
  the subset of $I$ containing only the paths $P_j$ with $j<
  o$ ($j>o$). Finally, let
  $T'_{<o}$ be the subtree of $T'$ restricted to the paths $P_j$ of $X$ with
  $j< o$. Note that at position $J$, $T_{<o}'$ is connected and 
  the paths in $I_{<o}$ are intact also in $T_{<o}'$.
  Then, the connector chooses the assignment $\beta : \{v\} \rightarrow
  \{0,1\}$ such that:
  $$
  \beta(v)=
  \begin{cases}
    \tau(v) & \quad |I_{<o}| < \sqrt{\ell_{i-1}}, \\
    1 & \quad |I_{<o}| \ge \sqrt{\ell_{i-1}}\text{ and }+\in \alpha,\\
    0 & \quad \text{otherwise}.
  \end{cases}
  $$

  As we will show below, $\beta$ is defined in such a manner that
  the position $J'=J[\beta]$ reached after the next round of the game
  $\game$ contains a subtree $T''$ of $T'$
  containing at least $\ell_i=\sqrt{\ell_{i-1}}/8$ paths that are
  intact in $J'$ and at most $z_i=z_{i-1}+2$ broken paths,
  which completes the proof since the
  connector can now chose the component of $J'$ containing $T''$
  to fulfill the induction invariant.
  We distinguish the following cases; refer also to Figure~\ref{fig:sepobslb} for an illustration of the two cases.
  
  \mpara{Case 1: $\boldsymbol{|I_{<o}| \ge \sqrt{\ell_{i-1}}}$.}
  We will show that $T''$ can be obtained as a subtree of $T_{<o}'$.

  Note first that all clauses $b_j$ with $j<o$ that are \bad{} in $J$
  are also \bad{} in $J'$. This is because $v \notin V_{j}$ (because
  $j<o$ and $v \notin V_{o-1}$) and
  therefore $v$ cannot $\alpha$-occur in $b_j$, which implies
  that $b_j$ remains \bad{} and not satisfied after setting $v$ to
  $\beta(v)$.
  
  The tree $T_{<o}'$ in $J$ may decompose into multiple components in $J'$.
  We will argue that one of these components contains many intact
  paths and only at most two more broken paths than $T_{<o}'$.
  Since the \bad{} clauses of an intact path remain \bad{} in $J'$,
  the only way in which an intact path can become broken is if parts
  of the path get removed, i.e., either $v$ or clauses satisfied by
  setting $v$ to $\beta(v)$.

  If $\beta(v)=1$ then $+\in\alpha$.
  If $\beta(v)=0$ then $+\not\in\alpha$, and since $\alpha\neq\emptyset$, then $-\in\alpha$.
  Thus, in $J'=J[\beta]$, the only elements that are removed are the variable $v$
  as well as clauses in which $v$ $\alpha$-occurs.
  By Lemma~\ref{lem:sep-obs-prop} (C3), $v$
  $\alpha$-occurs in at most two clauses of $T_{<o}'$ and
  because of (C5) those clauses have degree at most two in
  $T_{<o}'$. Therefore, setting $v$ to $\beta(v)$ removes at most
  two clauses from $T_{<o}'$, each of which having degree at most two.
  Moreover, according to Lemma~\ref{lem:sep-obs-prop} (C6), $v$ itself
  has degree at most three in $T_{<o}'$.
  This implies that setting $v$ to $\beta(v)$ splits $T_{<o}'$ into at most
  $2\cdot 2+3=7$ components. 

  Moreover, because of
  Lemma~\ref{lem:sep-obs-prop} (C3), the at most two clauses of $T_{<o}'$
  in which $v$ $\alpha$-occurs are located on the same path $P_j$.
  Therefore, at most two paths that are complete in $T_{<o}'$, i.e., the
  path $P_j$ and the at most one path containing $v$, can become broken. 
  Therefore, there is a component of
  $J'$ that contains a subtree of $T_{<o}'$ that contains at least
  $|I_{<o}|/7-2 \ge \sqrt{\ell_{i-1}}/7-2$ intact paths and at most
  $z_{i-1}+2\leq 2i= z_i$ broken paths of $X$.
  Note that $\sqrt{\ell_{i-1}} \ge \ell_d \ge 14^2+2d \ge 14^2$ 
  and therefore $\sqrt{\ell_{i-1}}/7-2 \geq \sqrt{\ell_{i-1}}/8 = \ell_i$.

  \begin{figure}
      \begin{center}
          \includegraphics[width=\textwidth]{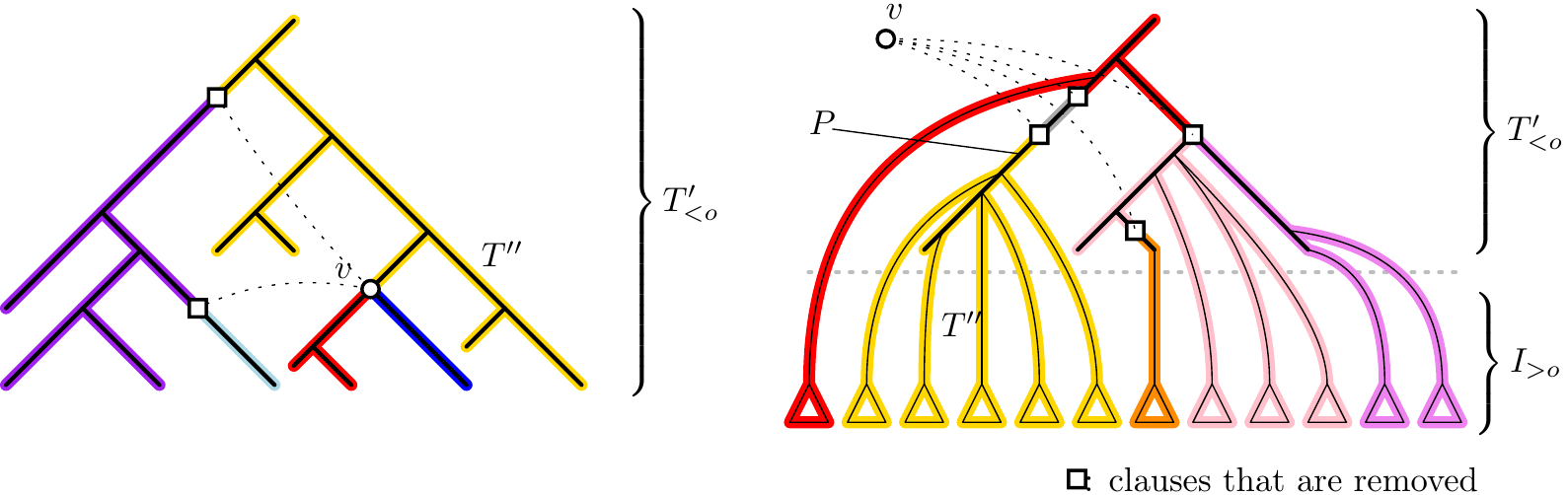}
        \end{center}
      \caption{
          Left: Case 1. The set $I_{<o}$ is large. Assigning $v$ to $\beta(v)$
          decomposes the tree $T'_{<o}$ into at most seven components. The largest component $T''$ is still large.
          Right: Case 2. The set $I_{<o}$ is small. There is a path $P$ to which many paths are weakly attached, forming a tree $T_P$.
          Assigning $v$ to $\beta(v)$ splits $T_P$ in at most three parts. The largest component $T''$ of $T_P$ is still large.
        }
        \label{fig:sepobslb}
  \end{figure}

  \mpara{Case 2: $\boldsymbol{|I_{<o}| < \sqrt{\ell_{i-1}}}$.}
  This means $\beta(v)=\tau(v)$. 
  In this case, we will build the subtree $T''$ 
  by picking only one path from $T_{<o}'$ and the remaining paths from
  $P_{o+1},\dotsc,P_\ell$. Let $A$ be the set of all paths of $X$ that
  are active in $T'$ and let $A_{>o}$ ($A_{<o}$) be the subset of $A$
  containing only the paths $P_j$ with $j>o$ ($j<o$).
  We say that a path $P_a$ of $X$ is \emph{attached} to a path $P_b$  of $X$
  if $a>b$, $V(P_a)\cap V(P_b)\neq \emptyset$ and there is no $b'< b$
  with $V(P_a)\cap V(P_{b'})\neq \emptyset$. We say that a path $P_a$ in
  $A_{>o}$ is \emph{weakly attached} to a path $P_b$ in $A_{<o}$ if
  either:
  \begin{itemize}
  \item $P_a$ is attached to $P_b$ or
  \item $P_a$ is attached to a path $P_c$ in $A_{>o}$ that is in
    turn attached to $P_b$.
  \end{itemize}
  Note that because $T'$ is a tree, every path in $A_{>o}$ is weakly
  attached to exactly one path in $A_{<o}$. Moreover, for the same
  reason any path in $A_{<o}$ together with all paths in $A_{>o}$
  that are weakly attached to it forms a subtree of $T'$.

  Therefore, there is a path $P$ in $A_{<o}$ such
  that at least $|I_{>o}|/|A_{<o}|$ paths in $I_{>o}$ are
  weakly attached to $P$. Moreover, the union $T_P$ of $P$ and all
  paths in $A_{>o}$ that are weakly attached to $P$ is a subtree of $T'$. 
  Note that $T_P$ has at least $|I_{>o}|/|A_{<o}|$
  paths that are intact in $T_P$
  and at most $z_{i-1}$ paths that are broken in $T_P$ at position $J$.
  Since $|I_{<o}| \ge \sqrt{\ell_{i-1}} \ge  \ell_d = 14^2+2d \ge 2d$ and $z_{i-1} \le z_d = 2d$,
  it holds that $|I_{<o}|+z_{i-1} \le 2|I_{<o}|$. Therefore,
  \iflong%
  \[\begin{array}{rl}
      |I_{>o}|/|A_{<o}| & \geq (\ell_{i-1}-|I_{<o}|)/(|I_{<o}|+z_{i-1})\\
                         & \geq (\ell_{i-1}-|I_{<o}|)/(2|I_{<o}|)\\
                         & \geq (\ell_{i-1})/(2\sqrt{\ell_{i-1}}) - 1/2\\
                         & \geq \sqrt{\ell_{i-1}}/2 -  1/2 \\
                         & = 8\ell_i/2 -  1 \geq 3 \ell_i.
  \end{array}\]
  \fi%
  \ifshort%
  \begin{multline*}
      |I_{>o}|/|A_{<o}| \geq (\ell_{i-1}-|I_{<o}|)/(|I_{<o}|+z_{i-1})
                         \geq (\ell_{i-1}-|I_{<o}|)/(2|I_{<o}|)
                         \geq \\ (\ell_{i-1})/(2\sqrt{\ell_{i-1}}) - 1/2
                         \geq \sqrt{\ell_{i-1}}/2 -  1/2
                         = 8\ell_i/2 -  1 \geq 3 \ell_i.
  \end{multline*}
  \fi

  Because $\beta(v)=\tau(v)$, all paths $P_j$ with $j>o$ that are
  active in $T_P$ are still contained in $J'$ and
  moreover if $P_j$ is intact in $J$, then it is still intact in $J'$. 
  Moreover, because of
  Lemma~\ref{lem:sep-obs-prop} (C2), $v$ occurs in at most two
  clauses of $P$ and because $\beta(v)=\tau(v)$ all paths
  $P_{o+1},\cdots, P_\ell$ that are attached to $P$ are still attached
  to $P$ after setting $v$ to $\beta(v)$. It follows that setting $v$
  to $\beta(v)$ removes at most two clauses and at most one variable
  (i.e., the variable $v$) from $P$ and also from $T_P$. Therefore,
  $J'=J[\beta]$ contains a component that contains a subtree $T''$ of $T_P$
  with at least $3\ell_i/3 = \ell_i$ paths
  that are intact in $T''$ and at most $z_{i-1}+1 \le z_i$ paths that are
  broken in $T''$.

\end{proof}

\textEnd[category=sobstructionTrees]{
\section{Basic Properties of Obstruction Trees}\label{sec:obstructionTrees}


In this section, we prove with \Cref{lem:obstructToBackdoor} that obstruction trees yield indeed a lower bound on the backdoor depth.
\begin{lemma}\label{lem:undoAssignment}
  Let $\satf \in \CNF$ and \CpnsBeAClass.
  Let $\beta$ be a partial assignment of the variables in $\satf$
  and $T$ be a $\CCC$-obstruction tree of depth $d$ in $\satf[\beta]$.
  Then, $T$ is also a $\CCC$-obstruction tree of depth $d$ in $\satf$.
\end{lemma}
\begin{proof}
  We use induction on $d$. If $d=0$, then there is a \bad
  clause $c$ of $\satf[\beta]$ such $T=\{c\}$. Therefore, $c$
  is also a \bad clause in $\satf$ and $T$ is a
  $\CCC$-obstruction tree of depth $0$ in $\satf$.

  Towards showing the induction step, let $d>0$. Then there is a
  $\CCC$-obstruction tree $T_1$ in $\satf[\beta]$ of depth $d-1$, an
  assignment $\beta'$ of the variables in $\satf[\beta]$ and a
  $\CCC$-obstruction tree $T_2$ in $\satf[\beta \cup \beta']$ of
  depth $d-1$ such that no variable of $\satf[\beta \cup \beta']$
  occurs both in a clause of $T_1$ and $T_2$. Moreover, there
  is a path $P$ connecting $T_1$ and $T_2$ in $\satf$. Because of the
  induction hypothesis, $T_1$ is a $\CCC$-obstruction tree in $\satf$
  of depth $d-1$. Therefore, $T=T_1\cup T_2\cup \var(P)\cup P$ is a
  $\CCC$-obstruction tree in $\satf$ of depth $d$,

\end{proof}

\begin{lemma}\label{lem:component}
    Let $\satf \in \CNF$ and \CpnsBeAClass.
  A set $T$ is a $\CCC$-obstruction tree of depth $d$ in $\satf$ if and only if
  there is a component $\satf'$ of $\satf$ such that~$T$ is a $\CCC$-obstruction tree of depth $d$ in $\satf'$.
\end{lemma}
\begin{proof}
  This follows because all variables and clauses belonging to $T$
  induce a connected subgraph of $\satf$.
\end{proof}
\begin{lemma}\label{lem:variableOutsideT}
    Let $\satf \in \CNF$ and \CpnsBeAClass.
    Let $T$ be a $\CCC$-obstruction tree of depth $d$ in $\satf$, $v \in
    \var(\satf)$ be a variable not occurring in any clause of $T$,
    and $\tau$ be an assignment to that variable.
    Then, $T$ is a $\CCC$-obstruction tree of depth $d$ in $\satf[\tau]$.
\end{lemma}
\begin{proof}
  Because $v$ does not appear in any clause of $T$, all clauses
  of $T$ in $\satf$ are still present in $\satf[\tau]$ and contain the
  same literals as in $\satf$; this also implies that every such
  clause is \good (\bad) in $\satf$ if and only if
  it is in $\satf[\tau]$. Moreover, because every variable of $T$ 
  occurs in some clause of $T$, it also follows that $v$ is not
  contained as a variable in $T$. Therefore, $T$ has the same
  variables and clauses in $\satf$ as in $\satf[\tau]$ and
  moreover all clauses remain the same, which shows the lemma.
\end{proof}
We are now ready to show the most crucial property of obstruction
trees, namely, that they can be used to obtain lower bounds for the
backdoor depth of a CNF formula.

\begin{proof}[Proof of \Cref{lem:obstructToBackdoor}]
  Assume there exists a $\CCC$-obstruction tree of depth $d$ in $\satf$.
  We will show that the connector has a strategy for the game $\game$
  to reach within $i\leq d$ rounds a position $J$ such
  that there is a $\CCC$-obstruction tree in $J$ of depth $d-i$. 
  This allows the connector to reach after $d$ rounds a position that contains
  a $\CCC$-obstruction tree of depth $0$, i.e., a \bad clause.
  Thus, the splitter has no strategy to win the game after at most $d$ rounds
  and by \Cref{lem:strategyEqualsBackdoor}, the statement of this lemma is proven.

  We show the claim by induction on $i$. The claim trivially holds for
  $i=0$. So suppose that $i>0$ and let $J$ be the position
  reached by the connector after $i-1$ rounds. Then, by the induction
  hypothesis, $J$ contains a $\CCC$-obstruction tree $T$ of depth at
  least $d-i+1>1$. 

  Since the depth is at least one,
  there further exist
  a $\CCC$-obstruction tree $T_1$ of depth $d-i$ in $J$,
  a partial assignment $\beta$ of the variables in $J$, and
  a $\CCC$-obstruction tree $T_2$ of depth $d-i$ in $J[\beta]$ such that
  that no variable $v \in \var(J[\beta])$ occurs both in a clause of $T_1$ and $T_2$.
  Now, let $v$ be the next variable chosen by the
  splitter. We distinguish the following cases:
  \begin{enumerate}
  \item 
    Assume $v$ does not occur in any clause of $T_1$.
    Then, as the connector, we assign $v$ an arbitrary value. Let $\tau$ be this assignment.
    By \Cref{lem:variableOutsideT}, $T_1$ is a $\CCC$-obstruction tree of depth $d-i$ in $J[\tau]$.
    We choose the connected component $J'$ containing $T_1$ in $J[\tau]$.
    By \Cref{lem:component} $T_1$ is also a $\CCC$-obstruction tree of depth $d-i$ in $J'$.
  \item
    Assume $v$ does not occur in any clause of $T_2$.
    By \Cref{lem:undoAssignment}, $T_2$ is a $\CCC$-obstruction tree of depth $d-i$ in $J$.
    We proceed analogously to the previous case, and reach a position
    in which $T_2$ is a $\CCC$-obstruction tree of depth $d-i$.
  \item
    Otherwise, $v$ occurs both in a clause of $T_1$ and
    $T_2$. Since $T$ is a $\CCC$-obstruction tree in~$J$, it
    follows by \Cref{def:obstructionTree} that $v \notin \var(J[\beta])$ and therefore $v \in \var(\beta)$.
    The connector can now choose the assignment $v=\beta(v)$ for $v$.
    Because of
    \Cref{lem:undoAssignment}, $T_2$ is a $\CCC$-obstruction tree of
    depth $d-i$ in $J[v=\beta(v)]$.
    Next, the connector chooses the component of $J[v=\beta(v)]$ containing $T_2$.
    By \Cref{lem:component}, $T_2$ is also a $\CCC$-obstruction tree in this component.
  \end{enumerate}%
\end{proof}

}

\section{Winning Strategies and Algorithms}\label{sec:find}

We are ready to present our algorithmic results.
Earlier, we discussed that separator obstructions are used to separate existing
obstruction trees from future obstruction trees.
As all obstruction trees are built only from shortest paths,
it is sufficient to derive a splitter-algorithm
that takes a shortest path $P$ and separates it from all future obstructions.
By reaching a position $J$ such that no variable in $\var(J)$ occurs in a clause of $P$,
we are guaranteed that all future obstructions are separated from $P$,
as future obstructions will only contain clauses and variables from $J$.

\begin{lemma}\label{lem:separate}
  Let \CpnsBeAClass.
  There exists a splitter-algorithm that implements a strategy
  to reach
  for each game $\game$, non-negative integer $d$, and shortest path $P$ between two \bad clauses in $\satf$
  within at most $(3s+d+1)\sepobsiz$ rounds either:
  %
  %
  \begin{enumerate}[1)]
  \item a winning position, or
  \item a position $J$ such that no variable in $\var(J)$ occurs in a clause of $P$, or
  \item a proof that the $\CCC$-backdoor depth of $\satf$ is at least $d$. 
  \end{enumerate}
  This algorithm takes at most $\bigoh(\|\satf\|)$ time per move.
\end{lemma}
\begin{proof}
If a clause of $F$ contains more than $d+s$ $\alpha$-literals, then this constitutes
by Lemma~\ref{lem:bc-arity} a proof that the $\CCC$-backdoor depth of
$\satf$ is at least $d$ and we archive case 3).
Thus, we can assume that every clause in every position of the game contains at most $d+s$ $\alpha$-literals.

  Let $\tuple{\seq{P_1,\dotsc,P_\ell},\tau}$ be a $\CCC$-separator
  obstruction for $\satf$ and let $\tau'$ be a sub-assignment of~$\tau$ assigning
  at least all variables in $V_{\ell-1}$. Then, we call
  $\tuple{\seq{P_1,\dotsc,P_\ell},\tau'}$ a \emph{partial $\CCC$-separator obstruction} for $\satf$.
  Consider the following splitter-algorithm, where
  we associate with each position $J$ of the game $\game$ a
  partial $\CCC$-separator obstruction $X(J)$ of the form $X(J)=\tuple{\seq{P_1,\dotsc,P_\ell},\tau_J}$ with $P_1=P$.
  We set $X(S)=\tuple{\seq{P},\emptyset}$ for the starting position $S$ of the game.

  Then, the splitter-algorithm does the following for a position $J$ in $\game$. 
  Let $X(J)=\tuple{\seq{P_1,\dotsc,P_\ell},\tau_J}$ and $V_i,B_i,T_i,T$ as in \Cref{def:separatorObstruction}.
  If there is at least one variable in $V_\ell\setminus V_{\ell-1}$ (where we set $V_0=\emptyset$) that has not
  yet been assigned by $\tau_J$, the splitter chooses any such
  variable. Otherwise, $X(J)$ is a $\CCC$-separator obstruction and
  we distinguish the following cases:
  \begin{enumerate}
  \item If there is a \bad clause in $J$ that has a path to some
      vertex of $T_\ell$, then let $b_{\ell+1}$ be a \bad clause that is
      closest to any vertex of $T_\ell$ in $J$ and let $P_{\ell+1}$ be a
    shortest path from $b_{\ell+1}$ to some vertex of $T_\ell$ in $J$. 
    Note that $\tuple{\seq{P_1,\dotsc,P_\ell,P_{\ell+1}},\tau_J}$ is 
    a partial $\CCC$-separator obstruction for $\satf$. The splitter
    now chooses any variable in $V_{\ell+1}\setminus V_{\ell}$ and
    assigns $X(J')=\tuple{\seq{P_1,\dotsc,P_\ell,P_{\ell+1}},\tau_{J'}}$
    for the position $J'$ resulting from this move.
  \item Otherwise, $X(J)$ can no longer be extended and either: (1)
    there is no \bad clause in $J$, in which case we reached a
    winning position (i.e., we achieved case 1)), or (2) every \bad clause of $J$ has no path
    to $T_\ell$, which implies that no variable of $J$ occurs
    in a clause of $T_\ell$ and therefore also of $P$ (i.e., we achieved case 2)).
  \end{enumerate}
  This completes the description of the splitter-algorithm.
  Moreover, if every play against the splitter-algorithm ends after
  at most $(3s+d+1)\sepobsiz$ rounds, every position is
  either of type~1) or type 2) and we are done.
  Otherwise, after playing for $(3s+d+1)\sepobsiz$ rounds we reach a position $J$.
  As stated at the beginning of the proof, every clause contains at most $d+s$ $\alpha$-literals
  and therefore, we obtain from Lemma~\ref{lem:sep-obs-prop} (C7) that $|V_{i+1}\setminus V_{i}| \le 3s+d+1$.
  This means that the size of the $\CCC$-separator obstruction increases by at least $1$ after at
  most $3s+d+1$ rounds.
  This means, at position $J$, reached after
  $(3s+d+1)\sepobsiz$ rounds,
  there is a partial $\CCC$-separator obstruction
  $X(J)$ of size at least $\sepobsiz$.
  By \Cref{lem:sep-obs-rbd}, this is a proof that $\satf$ has
  $\CCC$-backdoor depth at least $d$.

  Finally, the splitter-algorithm takes time at most $\bigoh(\|\satf\|)$
  per round since a \bad clause that is closest to the
  current $\CCC$-separator obstruction and the associated shortest
  path can be found using a simple breadth-first search.
\end{proof}

\textEnd[category=sfind]{
Since selecting more variables can only help the splitter in archiving
their goal, we immediately also get the following statement from \Cref{lem:separate}.
\begin{corollary}\label{cor:separate}
Consider \CpnsBeAClass,
a game $\game$ and a position $J'$ in this game, a non-negative
integer $d$ and shortest path $P$ between two \bad clauses in $\satf$.
There exists a splitter-algorithm that implements a strategy that continues the
game from position $J'$ and reaches within at most
$(3s+d+1)\sepobsiz$ rounds either:
\begin{enumerate}[1)]
  \item a winning position, or
  \item a position $J$ such that no variable in $\var(J)$ occurs in a clause of $P$, or
  \item a proof that the $\CCC$-backdoor depth of $\satf$ is at least $d$.
\end{enumerate}
This algorithm takes at most $\bigoh(\|\satf\|)$ time per move.
\end{corollary}

}

As described at the end of Section~\ref{sec:overview},
we can now construct in the following lemma
obstruction trees of growing size,
using the previous \ifshort lemma \fi\iflong corollary \fi to separate them
from potential future obstruction trees.

\begin{theoremOptional}{sfind}{lemma}\label{lem:main}
  Let \CpnsBeAClass.
  There is a splitter-algorithm that implements a strategy to reach
  for a game $\game$ and non-negative integers $i$, $d$ with $1 \le i \le d$ within at most
  $(2^{i}-1)(3s+d+1)\sepobsiz$ rounds either:
  \begin{enumerate}[1)]
        \item a winning position, or
        \item a position $J$ and a $\CCC$-obstruction tree $T$ of depth $i$ in $\satf$ such that
              no variable in $\var(J)$ occurs in a clause of $T$, or
        \item a proof that the $\CCC$-backdoor depth of $\satf$ is at least $d$.
  \end{enumerate}
  This algorithm takes at most $\bigoh(\|\satf\|)$ time per move.
\end{theoremOptional}
\begin{proofEnd}
    We will prove this lemma by induction over $i$.
    Our splitter-algorithm will try construct an obstruction tree of depth $i$
    by first using the induction hypothesis to build two obstruction trees $T_1$ and $T_2$ of depth $i-1$
    and then joining them together.
    After the construction of the first tree~$T_1$, we reach a position $J_1$ and 
    by our induction hypothesis no variable in $\var(J_1)$ occurs in a clause of $T_1$.
    This encapsulates the core idea behind our approach, as it
    means that $T_1$ is separated from all potential future obstruction trees $T_2$ that we build from position $J_1$.
    Therefore, we can compute the next tree $T_2$ in $J_1$ and join $T_1$ and $T_2$ together in accordance with \Cref{def:obstructionTree} by a path $P$.
    At last, we use \Cref{cor:separate} to also separate this path from all future obstructions.
    If at any point of this process we reach a winning position or a proof
    that the $\CCC$-backdoor depth of $\satf$ is at least $d$, we can stop.
    Let us now describe this approach in detail.

    For convenience, let $x = (3s+d+1)\sepobsiz$.
    We start our induction with $i=1$.
    If there is no \bad clause in $\satf$, then it is a winning position and we can stop.
    Assume there is exactly one \bad clause $c$ in $\satf$.
    By Lemma~\ref{lem:bc-arity}, if $c$ contains more then $d+s$
    $\alpha$-literals, we have a proof that the $\CCC$-backdoor depth
    of $\satf$ is at least $d$ and we archive case 3) of the lemma.
    On the other hand, if $c$ contains at most $d+s$ $\alpha$-literals,
    the splitter can obtain a winning
    position in $\game$ after at most $d+s\leq (2^{i}-1)x$ rounds by choosing a new variable
    $\alpha$-occurring in $c$ at every round.
    Assume there is more than one \bad clause in $\satf$.
    Thus, we pick \bad clauses $c_1$ and $c_2$ and
    compute a shortest path~$P$ between $c_1$ and $c_2$ in $\satf$.
    By \Cref{def:obstructionTree}, $T = \{c_1\} \cup \{c_2\} \cup \var(P) \cup P$ is a
    $\CCC$-obstruction tree of depth $1$ in~$\satf$.
    We then
    continue the game using \Cref{cor:separate} (for the path $P$) to reach a position
    $J'$ satisfying (1), (2), or (3) after at most $x\leq
    (2^{i}-1)x$ rounds, with each round taking at most $\bigoh(\|\satf\|)$ time. 
    
    We now assume the statement of this lemma to hold for $i-1$ and we show it also holds for~$i$.
    To this end, we start playing the game $\game$ according to the existing splitter-algorithm for $i-1$.
    If we reach (within at most $(2^{i-1}-1) x$ rounds) a winning position or a
    proof that the $\CCC$-backdoor depth of $\satf$ is
    at least $d$ then we are done.
    Assuming this is not the case, we reach a position $J_1$ and a
    $\CCC$-obstruction tree $T_1$ of depth $i-1$ in $\satf$ such that 
    no variable $v \in \var(J_1)$ occurs in a clause of $T_1$.

    We continue playing the game at position $J_1$ according to the existing splitter-algorithm for $\gameJOne$ and $i-1$.
    The $\CCC$-backdoor depth of $\satf$ is larger or equal to the $\CCC$-backdoor depth of $J_1$.
    Thus again (after at most $(2^{i-1}-1)x$ rounds) we either are done
    (because we reach a winning position or can conclude that the
    $\CCC$-backdoor depth of $J_1$ is at least $d$)
    or we reach a position $J_2$ and a $\CCC$-obstruction tree $T_2$ of depth $i-1$ in $J_1$ such that 
    no variable $v \in \var(J_2)$ occurs in a clause of $T_2$.

\begin{figure}
    \begin{center}
    \includegraphics{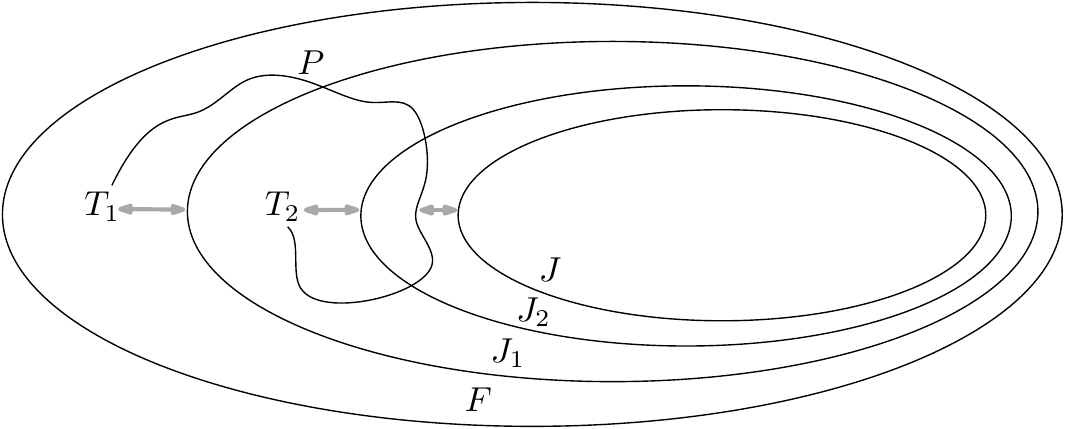}
    \end{center}
    \caption{Overview of the construction in \Cref{lem:main}.
    First, $T_1$ is chosen in $F$, yielding $J_1$.
    Then, $T_2$ is chosen in $J_1$, yielding $J_2$.
    In the end the connecting path $P$ is chosen yielding $J$.
    A gray doublesided arrow between a position $\widehat J$ and structure $\widehat T$ symbolizes
    that no variable $v \in \var(\widehat J)$ occurs in a clause of $\widehat T$.}
    \label{fig:main}
\end{figure}

    We pick two clauses $c_1 \in T_1$ and $c_2 \in T_2$ that are
    \bad in $\satf$ and compute a shortest path~$P$ between $c_1$ and $c_2$ in $\satf$.
    We now argue that $T = T_1 \cup T_2 \cup \var(P) \cup P$ is a
    $\CCC$-obstruction tree of depth $i$ in~$\satf$.
    Let $\beta = \tau_{J_1}$ be the assignment that assigns all the variables
    the splitter chose until reaching position $J_1$ to the value given by the connector.
    Note that $J_1$ is a connected component of $\satf[\beta]$.

    Since all variables and clauses belonging to $T_2$ induce a connected subgraph of $J_1$,
    $T_2$ is a $\CCC$-obstruction tree of depth $i-1$ not only in $J_1$, but also in $\satf[\beta]$.
    Let $v \in \var(\satf[\beta])$.
    We show that $v$ does not occur both in some clause of $T_1$ and of $T_2$.
    To this end, assume $v$ is contained in a clause of $T_2$.
    Since all clauses of $T_2$ are in $J_1$ and $J_1$ is a connected component of
    $\satf[\beta]$, we further have $v \in \var(J_1)$.
    On the other hand (as discussed earlier), no variable $v \in \var(J_1)$ is contained in a clause of $T_1$.
    By \Cref{def:obstructionTree}, $T = T_1 \cup T_2 \cup \var(P) \cup P$ is a
    $\CCC$-obstruction tree of depth $i$ in~$\satf$.

    We use \Cref{cor:separate} to continue playing the game at position $J_2$.
    Again, if we reach a winning position or a proof that the
    $\CCC$-backdoor depth of $\satf$ is at least $d$ we are done.
    So we focus on the third case that we reach (within at most $x$ rounds) 
    a position $J$ such that no variable $v \in \var(J)$ is contained in a clause of $P$.
    We know already that no variable $v \in \var(J_1)$ is contained in a clause of $T_1$
    and no variable $v \in \var(J_2)$ is contained in a clause of
    $T_2$.
    Since $\var(J) \subseteq \var(J_2) \subseteq \var(J_1)$,
    and $T = T_1 \cup T_2 \cup \var(P) \cup P$,
    we can conclude that
    no variable $v \in \var(J)$ is contained in a clause of $T$.

    In total, we played for $(2^{i-1}-1)x + (2^{i-1}-1)x + x = (2^{i}-1)x$ rounds.
    The splitter-algorithm in \Cref{cor:separate} takes at most $\bigoh(\|\satf\|)$ time per move.
    The same holds for the splitter-algorithm for $i-1$ that we use as a subroutine.
    Thus, the whole algorithm takes at most $\bigoh(\|\satf\|)$ time per move.
\end{proofEnd}

The main results now follow easily by combining
\Cref{lem:ConvertStrategyToBackdoor,lem:obstructToBackdoor,lem:main,lem:runtime}.
\begin{theoremOptional}{sfind}{theorem}\label{thm:main}
  Let \CpnsBeAClass. We can, for a
  given $\satf \in \CNF$ and a non-negative integer $d$, in time at most
  $2^{2^{2^{\bigoh(d)}}} \|\satf\|$ either
  \begin{enumerate}[1)]
  \item  compute a component $\CCC$-backdoor tree
    of $\satf$ of depth at most
    $2^{2^{\bigoh(d)}}$, or
  \item conclude that the $\CCC$-backdoor depth of
    $\satf$ is larger than $d$.
  \end{enumerate}
\end{theoremOptional}
\begin{proofEnd}
  We apply \Cref{lem:main} with its parameters $i$ and $d$ both set to $d+1$.
  An obstruction tree of depth $d$ is, according to \Cref{lem:obstructToBackdoor}, a proof that the
  backdoor depth is at least $d+1$, thus
  the output of the splitter-algorithm in \Cref{lem:main} after
  $2^{2^{\bigoh(d)}}$ rounds reduces to either a winning
  position, or a proof that the $\CCC$-backdoor depth of $\satf$ is
  at least $d+1$.  The algorithm takes at most $\bigoh(\|\satf\|)$ time per
  move.  The statement then follows from
  \Cref{lem:ConvertStrategyToBackdoor}.
\end{proofEnd}

\begin{theoremOptional}{sfind}{corollary}\label{cor:horn-dhorn-krom}
  Let $\CCC\in \{\Horn,\dHorn,\Krom\}$. The \SAT{} problem can be solved
  in linear time for any class of formulas of bounded $\CCC$\hy
  backdoor depth.
\end{theoremOptional}
\begin{proofEnd}
  Let $\satf \in \CNF$.
  We use
  \Cref{thm:main} to compute a component $\CCC$-backdoor tree for
  $\satf$ of
  depth at most $2^{2^{\bigoh(d)}}$ and then use \Cref{lem:runtime} to
  decide the satisfiability of
  $\satf$ in time $2^{2^{2^{\bigoh(d)}}}\size{\satf}$.
\end{proofEnd}

\textEnd[category=scomp]{
\section{Comparison with other Approaches}\label{sec:comp}

In this section, we compare the generality of backdoor depth with
other parameters that admit a fixed-parameter tractable solution of
the problem, as we have listed in the introduction. Our comparison is
based on the concept of domination~\cite{SamerSzeider10a}. For two
integer-valued parameters $p$ and $q$, we say that $p$
\emph{dominates} $q$ if every class of instances for which $q$ is
bounded, also $p$ is bounded; $p$ \emph{strictly dominates} $q$ if $p$
dominates $q$ but $q$ does not dominate $p$.  If $p$ and $q$ dominate
each other, they are \emph{domination equivalent}. If neither of them
dominates the other, they are \emph{domination orthogonal}.

\newcommand{\Chain}{\mathsf{Q}}
\newcommand{\ChaiN}{\mathsf{Q}'}

We first define two sequences of formulas that we will use for several
separation results below.  For any integer $d\geq 1$, the CNF
formula $\Chain_d$, consists of $n=3 \cdot 2^d-2$ clauses
$c_1,\dots,c_n$ over the variables $x_0,\dots,x_n$ and
$y_1,\dots,y_n$, where $c_i=\{x_{i-1},\neg y_i,x_i\}$.  The formula
$\ChaiN_d$ is defined similarly, except that the $y_i$ variables are
omitted, and hence $c_i=\{x_{i-1},x_i\}$. By construction, we have
 $\Chain_d \in \dHorn \setminus (\Horn \cup \Krom)$ and $\ChaiN_d \in \Krom \cap \dHorn \setminus \Horn$.

\begin{lemma}\label{lem:chain}
  $\srbd_\Horn(\Chain_d)=\srbd_\Krom(\Chain_d)=\srbd_\Horn(\ChaiN_d)=d$.
\end{lemma}
\begin{proof}
  We focus on showing that
  $\srbd_\CCC(\Chain_d)=\srbd_\CCC(\Chain_d)=d$ for
  $\CCC\in \{\Horn,\Krom\}$, the proof for $\srbd_\Horn(\ChaiN_d)=d$
  is similar.  We proceed by induction on $d$.  Since
  $\Chain_1=\{ c_0 \}$, the induction basis holds. Now assume $d>1$.
  Since $\Chain_d$ is connected, the root~$r$ of any component
  $\CCC$\hy backdoor tree for it is a variable node. Assume $r$
  branches on variable $x_i$. By the symmetry of $\Chain_n$, we can
  assume, w.l.o.g., that $i\geq n/2$.  We observe that
  $\Chain_d[x_i=1]=\Chain_d\setminus \{c_{i},c_{i+1}\}$, hence
  $\Chain_{d-1} \subseteq \Chain_d[x_i=1]$, where $\Chain_{d-1}$ has
  $3 \cdot 2^{d-1}-2$ clauses.  If $r$ branches on $y_i$, then
  $\Chain_d[x_i]\subseteq \Chain_d[y_i=0]$, and so
  $\Chain_{d-1} \subseteq \Chain_d[y_i=0]$ as well.  Consequently, in
  any case, at least one child of $r$ is labeled with a formula that
  contains $\Chain_{d-1}$. By induction hypothesis
  $\srbd_\CCC(\Chain_{d-1})=d-1$. Thus $\srbd_\CCC(\Chain_d)=d$ as
  claimed.
\end{proof}

\subsection{Backdoor size into \Horn, \dHorn, and \Krom}

Let $\CCC\subseteq \CNF$. The \emph{$\CCC$\hy backdoor size} of a
CNF formula $F$, denoted $\sbs_\CCC(F)$ , is the size of a smallest
(strong) $\CCC$\hy backdoor of $\satf$.  Nishimura et
al.~\cite{NishimuraRagdeSzeider04-informal} have shown that \SAT is
fixed-parameter tractable parameterized by $\CCC$\hy backdoor size
for $\CCC\in \{\Horn,\dHorn,\Krom\}$.
A class $\CCC$ of CNF formulas is
\emph{nontrivial} if $\CCC\neq \CNF$.

\begin{proposition}\label{pro:srbd-sbs}
  For every nontrivial class $\CCC$ of CNF formulas, $\srbd_\CCC$
  strictly dominates $\sbs_\CCC$.
\end{proposition}
\begin{proof}
  Clearly, $\srbd_\CCC$ dominates $\sbs_\CCC$, since if we have a
  $\CCC$\hy backdoor $B$ of some~$\satf\in \CNF$, we can build a backdoor
  tree with at most $2^\Card{B}$ variable nodes and depth
  $\leq \Card{B}$.

  To show that the domination is strict, let $F\in \CNF\setminus \CCC$
  (it exists, since $\CCC$ is nontrivial) and let $F_1,\dots,F_n$ be
  variable-disjoint copies of $F$.  Since $F\notin \CCC$, we have
  $\sbs_\CCC(F)\geq \srbd_\CCC(F) >0$. Let $\sbs_\CCC(F)=s$ and
  $\srbd_\CCC(F)=d$.  Now $\sbs_\CCC(F_n)=ns$, so we can choose $n$ to
  make $\sbs_\CCC(F_n)$ arbitrarily large.  However,
  $\srbd_\CCC(F_n)=d$ remains bounded, since in a backdoor tree we can
  use a component node at the root which branches into (at least) $n$
  components, each of $\CCC$\hy backdoor depth at most~$d$.
\end{proof}

\subsection{Number of leaves of backdoor trees into  \Horn, \dHorn, and \Krom}
Backdoor trees, introduced by Samer and Szeider
\cite{SamerSzeider08b,OrdyniakSchidlerSzeider20}, are a special case
of component backdoor trees as defined in Section~\ref{sec:srbd}.
A \emph{$\CCC$\hy backdoor tree} for $F\in\CNF$ is a component
$\CCC$\hy backdoor tree for $F$ without component nodes.
For a base class $\CCC$ and $F\in \CNF$, let $\bdt_\CCC(\satf)$ denote
the smallest number of leaves of any $\CCC$\hy backdoor tree (recall
the definition in \Cref{sec:srbd}).  \SAT is fixed-parameter tractable
parameterized by $\bdt_\CCC(F)$~\cite{OrdyniakSchidlerSzeider20}.  Since
$\sbs_\CCC(\satf)+1 \leq \bdt_\CCC(\satf) \leq 2^{\sbs_\CCC(\satf)}$,
$\bdt_\CCC$ and $\sbs_\CCC$ are domination equivalent
\cite{SamerSzeider08b}. Hence we have the following corollary to Proposition~\ref{pro:srbd-sbs}.

\begin{corollary}
  For every nontrivial class $\CCC$ of CNF formulas, $\srbd_\CCC$
  strictly dominates $\bdt_\CCC$.
\end{corollary}

\subsection{Backdoor depth into \EMPTY}

\MSV considered the base class $\EMPTY=\{\emptyset, \{\emptyset\}\}$
and showed that \SAT is fixed-parameter tractable parameterized by
$\srbd_\EMPTY$.
\begin{proposition}\label{pro:dom-empty}
  $\srbd_\CCC$ strictly dominates $\srbd_\EMPTY$ for
  $\CCC\in \{\Horn,\dHorn,\Krom\}$.
\end{proposition}
\begin{proof}
  It suffices to consider $\CCC\in \{\dHorn,\Krom\}$, since the cases
  $\dHorn$ and $\Horn$ are symmetric.  Since
  $\EMPTY\subseteq \dHorn\cap\Krom$, $\srbd_\CCC$ dominates
  $\srbd_\EMPTY$. To show that the domination is strict, consider the
  CNF formulas $\Chain_d$ and $\Chain_d'$ from above, with
  $\srbd_\Horn(\Chain_d)=\srbd_\Horn(\Chain_d')=d$ by
  Proposition~\ref{lem:chain}.  Since $\EMPTY\subseteq \Horn$, 
  $\srbd_\EMPTY(\Chain_d)=\srbd_\EMPTY(\ChaiN_d)\geq d$.  However,
  $\srbd_\dHorn(\Chain_d)=\srbd_\Krom(\ChaiN_d)=0$ since
  $\Chain_d\in \dHorn$ and $\ChaiN_d\in \Krom$.
  \end{proof}

\subsection{Backdoor treewidth into  \Horn, \dHorn, and \Krom}

\emph{Backdoor treewidth} is another general parameter for \SAT (and
CSP) defined with respect to a base
class~$\CCC$~\cite{GanianRamanujanSzeider17,GanianRamanujanSzeider17b}.
\newcommand{\torso}{\TTT}  
  Let $\satf$ be a CSP instance and $X\subseteq V(\satf)$ a subset of
  its variables.  The \emph{torso graph} of $\satf$ with respect to $X$, denoted
  $\torso_\satf(X)$, has as vertices the variables in $X$ and contains
  an edge $\{x,y\}$ if and only if  $x$ and $y$ appear
  together in the scope of a constraint of $\satf$ or $x$ and $y$ are in the
  same connected component of the graph obtained from the incidence
  graph of $\satf$ after deleting all the variables in~$X$.
  Let $G=(V,E)$ be a graph. A tree decomposition of $G$ is a pair
  $(T,\XXX)$, $\XXX=\{X_{t}\}_{t\in V(T)}$, where $T$ is a tree and
  $\XXX$ is a collection of subsets of $V$ such that: (i)~for each
  edge $\{u,v \}\in E$ there exists a node $t$ of $T$ such that
  $\{u,v\} \subseteq X_{t}$, and (ii)~for each $v\in V$, the set
  $\{t\mid v\in X_{t}\}$ induces in $T$ a nonempty connected
  subtree. The width of $(T,\XXX)$ is equal to
  $\max\{|X_t|-1\mid {t\in V(T)}\}$ and the \emph{treewidth} of $G$ is
  the minimum width over all tree decompositions of $G$.
  Let $\satf \in \CNF$ and $X$ a $\CCC$\hy backdoor set of
  $\satf$. The \emph{torso treewidth} of $X$ is the treewidth of the
  torso graph $\torso_F(X)$, and the \emph{$\CCC$\hy backdoor
    treewidth} of $\satf$, denoted $\bdtw_\CCC(\satf)$, is the
  smallest torso treewidth over all $\CCC$\hy backdoor sets of
  $\satf$.  It is known that $\bdtw_\CCC$ dominates the treewidth of
  the formula's primal graph and
  $\sbs_\CCC$~\cite{GanianRamanujanSzeider17,GanianRamanujanSzeider17b}.

  \begin{proposition}\label{pro:orthogonal}
    For $\CCC\in \{\Horn,\dHorn,\Krom\}$, the parameters $\srbd_\CCC$
    and $\bdtw_\CCC$ are domination orthogonal.
\end{proposition}
\begin{proof}
  Since the cases $\dHorn$ and  $\Horn$ are symmetric, we may assume
  that  $\CCC\in \{\Horn$, $\Krom\}$.

  First we construct a sequence of formulas of constant $\srbd_\CCC$
  but unbounded $\bdtw_\CCC$ (similar constructions have been used
  before~\cite{GaspersSzeider13,MaehlmannSiebertzVigny21}).  For any
  $n\geq 2$ we construct a CNF formula $F_n$ based on an $n\times n$
  grid graph $G_n$ whose edges are oriented arbitrarily. We take a
  special variable $x$. For each vertex $v\in V(G_n)$ we introduce
  three variables $y_v^i$, $1\leq i \leq 3$, and for each oriented
  edge $(u,v) \in E(G_n)$ we introduce a variable $z_{u,v}$ and the
  clauses $c_{u,v}=\{x,y_u^1,y_u^2,y_u^3,z_{u,v}\}$ and
  $d_{u,v}=\{\neg x,\neg y_v^1,\neg y_v^2,\neg y_v^3,\neg
  z_{u,v}\}$. Let $X$ be a $\CCC$\hy backdoor of~$F_n$.  We observe
  that for each $v\in V(G_n)$, $X$ must contain at least one of the
  variables $y_v^1,y_v^2,y_v^3$ since, otherwise, $F[X\mapsto 0]$
  would contain a clause (a subset of $c_{u,v}$) that is neither Horn
  nor Krom. W.l.o.g, we assume $y_v^1\in X$ for all $v\in V(G_n)$.
  Furthermore, for each edge $(u,v) \in E(G_n)$, the torso graph
  $\torso_F(X)$ will contain the edge $\{y_u^1,y_v^1\}$ or the
  edges $\{y_u^1z_{u,v}\}$ and $\{z_{u,v},y_v^1\}$. Consequently,
  $\torso_F(X)$ has a subgraph that is isomorphic to a subdivision
  of~$G_n$. Since the treewidth of $G_n$ is $n$, and subdividing edges
  does not change the treewidth, we conclude that
  $\bdtw_\CCC(F_n)\geq n$.

  It remains to show that $\srbd_\CCC$ does not dominate $\bdtw_\CCC$.
  Consider again the formula~$\Chain_d$ from above with
  $\srbd_\CCC(\Chain_d)=d$ by Lemma~\ref{lem:chain}. The set
  $X=\{x_0,\dots,x_n\}$, $n=3 \cdot 2^d -2$, is a $\CCC$\hy backdoor
  of $F_d$. The torso graph $\torso_X(F_d)$ is a path of length $n$,
  which has treewidth~$1$. Thus $\bdtw_\CCC(F_d)=1$.
  \end{proof}

\subsection{Backdoor size  into  heterogeneous base classes based on \\
   \Horn, \dHorn, and \Krom}\label{subs:hetero}
 Consider the possibility that different assignments to the backdoor
 variables move the formula into different base classes, e.g., for
 $B\subseteq \var(F)$ and $\tau,\tau'\rightarrow \{0,1\}$, we have
 $F[\tau]\in \Horn$ and $F[\tau']\in \Krom$.  We can see such a
 backdoor $B$ as a $\CCC$\hy backdoor for
 $\CCC=\Horn \cup \Krom$\footnote{Observe that for
   $F\in \Horn \cup \Krom$, either all of $F$'s clauses are Horn
   clauses, or all its clauses are Krom clauses, but $F$ does not
   contain a mixture of both (except for clauses that are both Horn
   and Krom)}.  Gaspers et
 al.~\cite{GaspersMisraOrdyniakSzeiderZivny17} have shown that
 $\srbd_{\Horn \cup \Krom}$ strictly dominates $\srbd_{\Horn}$ and
 $\srbd_{\Krom}$, They also showed that computing
 $\srbd_{\Horn \cup \Krom}$ is FPT, but computing
 $\srbd_{\Horn \cup \dHorn}$ is W[2]-hard (i.e., unlikely to be FPT).

\begin{proposition}\label{pro:orthogonal-h}
  For any two different classes
  $\CCC \neq \CCC' \in \{\Horn,\dHorn,\Krom\}$,
  $\sbs_{\CCC \cup \CCC'}$ and $\srbd_\CCC$ are domination orthogonal.
\end{proposition}
\begin{proof}
  Let $\CCC \neq \CCC' \in \{\Horn,\dHorn,\Krom\}$.
  
  With the same padding argument as used in the proof of
  Proposition~\ref{pro:srbd-sbs}, we can show that there are formulas
  of constant $\srbd_\CCC$ and arbitrarily large
  $\sbs_{\CCC \cup \CCC'}$.

\newcommand{\NChain}{\overline{\Chain}}
\newcommand{\NChaiN}{\overline{\ChaiN}}

For the converse direction, we utilize the formulas $\Chain$ and
$\ChaiN$ from above, as well as the formulas $\NChain$ and $\NChaiN$
obtained from $\Chain$ and $\ChaiN$, respectively, by flipping all
literals to the opposite polarity.  For any pair $\CCC,\CCC'$, we can
choose formulas $F_d,F_d' \in \{\Chain$, $\ChaiN$, $\NChain$,
$\NChaiN\}$, such that $F_d \in \CCC\setminus \CCC'$ and
$F_d' \in \CCC'\setminus \CCC$.  By Lemma~\ref{lem:chain} (applied
directly or via symmetry) we have
$\srbd_{\CCC'}(F_d)=\srbd_{\CCC}(F_d')=d$.  We may assume that $F_d$
and $F_d'$ have disjoint sets of variables. We pick a new variable $z$
and add it positively to all clauses of $F_d$, and negatively to all
clauses of $F_d'$, obtaining the formulas $G_d$ and $G_d'$,
respectively. Let $G^*_d=G_d \cup G_d'$.  On the one hand,
$\srbd_\CCC(G^*_d)\geq d$, since any component $\CCC$\hy backdoor tree
of $G^*_d$ must contain a subtree which induces a $\CCC$\hy backdoor
tree of $F_d'$, whose depth is $\geq \srbd_{\CCC}(F_d')=d$.  On the
other hand, $\sbs_{\CCC \cup \CCC'}(G^*)=1$, since
$G^*_d[z=0]=F_d\in \CCC \subseteq \CCC \cup \CCC'$ and
$G^*_d[z=1]=F_d'\in \CCC' \subseteq \CCC \cup \CCC'$.
\end{proof}

\subsection{Backdoor size into scattered base classes based on  \Horn, \dHorn, and \Krom}\label{subs:scattered}

For base classes $\CCC_1,\dots,\CCC_r\subseteq \CNF$, let
$\CCC_1 \oplus \dots \oplus \CCC_r\subseteq \CNF$ denote the class of
CNF formulas~$F$ with the property that each $F'\in \Conn(F)$ belongs
to $\CCC_i$ for some $i\in \{1,\dots,r\}$. For example, if
$F\in \Horn \oplus \Krom$, then each connected component of $F$ is
either Horn or Krom. Such \emph{scattered} base classes where
introduced by Ganian et al.~\cite{GanianRamanujanSzeider17a} for
constraint satisfaction, but the concept naturally extends to \SAT.

\begin{proposition}\label{pro:orthogonal-s}
  For any two different classes
  $\CCC \neq \CCC' \in \{\Horn,\dHorn,\Krom\}$,
  $\sbs_{\CCC \oplus \CCC'}$ and $\srbd_\CCC$ are domination orthogonal.
\end{proposition}
\begin{proof}
  Let $\CCC \neq \CCC' \in \{\Horn,\dHorn,\Krom\}$.  Again,
  we can use the padding argument from the proof of
  Proposition~\ref{pro:srbd-sbs} to show that there are formulas of constant
  $\srbd_\CCC$ and arbitrarily large $\sbs_{\CCC \oplus \CCC'}$.  For
  the converse direction, we argue as in the proof of
  Proposition~\ref{pro:orthogonal-s} that there are variable-disjoint formulas
  $F_d \in \CCC\setminus \CCC'$ and $F_d' \in \CCC'\setminus \CCC$
  with $\srbd_{\CCC'}(F_d)=\srbd_{\CCC}(F_d')=d$.
  For $F^*_d=F_d \cup F_d'$ we have  $\srbd_\CCC(F^*_d)\geq d$ but 
  $\sbs_{\CCC \oplus \CCC'}(F^*)=0$ since $F^*_d \in \CCC \oplus \CCC'$.
\end{proof}

\subsection{Deletion backdoor size  into \QHorn}\label{ssec:qhorn}
A CNF formula $F$ is \emph{quadratic Horn} if there exist a mapping
$f:\var(F)\rightarrow [0,1]$ such that for each clause $c\in F$ we
have
$\sum_{x\in C\cap \var(C)}f(x)+\sum_{\neg x \in C \setminus
  \var(C)}1-f(x)\leq 1$. The class $\QHorn$ of all quadratic Horn
formulas properly contains $\Horn$, $\dHorn$, and $\Krom$ (which can
be seen by taking $f$ to be the constant mapping to 1, to 0, and to
$1/2$, respectively). Satisfiability and recognition of quadratic Horn
formulas can be decided in polynomial time
\cite{BorosCramaHammer90,BorosHammerSun94}.  Gaspers et
al.~\cite{GaspersOrdyniakRamanujanSaurabhSzeider16} considered
\emph{deletion} $\QHorn$\hy backdoors, since computing $\sbs_\QHorn$
is W[2]-hard.  A deletion $\QHorn$\hy backdoor of $F\in \CNF$ is a set
$B \subseteq \var(F)$ such that
$F-B:= \SB c\setminus (B\cup \overline B) \SM c\in F\SE \in
\QHorn$. We denote the size of a smallest deletion $\QHorn$\hy
backdoor of $F$, i.e., the deletion $\QHorn$\hy backdoor size of $F$,
by $\dbs_\QHorn(F)$. Ramanujan and Saurabh \cite{RamanujanSaurabh17}
showed that computing $\dbs_\QHorn(F)$ is fixed-parameter tractable,
improving upon an FPT-approximation result by Gaspers et
al.~\cite{GaspersOrdyniakRamanujanSaurabhSzeider16}.  Every deletion
$\QHorn$\hy backdoor is a $\QHorn$\hy backdoor, but the
converse does not hold. Hence $\sbs_\QHorn$ strictly dominates
$\dbs_\QHorn$.

\begin{proposition}\label{pro:orthogonal-q}
  For $\CCC\in \{\Horn,\dHorn,\Krom\}$, the parameters $\srbd_\CCC$
  and $\dbs_\QHorn$ are domination orthogonal.
\end{proposition}
\begin{proof}
  The padding argument from the proof of
  Proposition~\ref{pro:srbd-sbs} shows that there are formulas of
  constant $\srbd_\CCC$ and arbitrarily large $\sbs_\QHorn$, and
  therefore of arbitrarily large $\dbs_\QHorn$.

  
  For the reverse direction, consider the formula $\Chain_d$. By
  Lemma~\ref{lem:chain}, $\srbd_\CCC(\Chain_d)=d$ for
  $\CCC \in \{\Horn, \Krom\}$. By construction,
  $\Chain_d\in \dHorn \subseteq \QHorn$, thus
  $\dbs_\QHorn(\Chain_d)=0$.  Hence neither $\srbd_\Horn$ nor
  $\srbd_\Krom$ dominates $\dbs_\QHorn$. A symmetric argument shows
  that $\srbd_\dHorn$ does not dominate $\dbs_\QHorn$.
 \end{proof}

\subsection{Backdoor size into bounded incidence
  treewidth}
The \emph{incidence treewidth} of a formula $F$ is the treewidth of
its incidence graph.  Let $\WWW_t$ denote the class of all CNF
formulas of incidence treewidth $\leq t$.  \SAT can be solved in
polynomial time for formulas in $\WWW_t$
\cite{SamerSzeider10,Szeider04b}.  Gaspers and
Szeider~\cite{GaspersSzeider13} showed that for every constant
$t\geq 1$, $\sbs_{\WWW_t}$ can be FPT approximated. Hence \SAT is FPT
parameterized by $\sbs_{\WWW_t}$.

\begin{proposition}\label{pro:orthogonal-t}
  For $\CCC\in \{\Horn,\dHorn,\Krom\}$ and every $t\geq 1$, the
  parameters $\srbd_\CCC$ and $\sbs_{\WWW_t}$ are domination
  orthogonal.
\end{proposition}
\begin{proof}
  Let $t\geq 1$ be a constant and $\CCC\in \{\Horn,\Krom\}$; the cases
  $\dHorn$ and $\Horn$ are symmetric.

  Consider the $(t+1)\times (t+1)$ grid graph $G_{t+1}$.  We construct the
  CNF formula $T$ by introducing a variable $x_v$ for every
  $v\in V(G_{t+1})$ and a clause $\{\neg x_u, \neg x_v\}$ for every
  edge $uv\in E(G_{t+1})$. Let $F_n$ be the formula consisting of $n$
  variable-disjoint copies of $T$.  By construction, $F_n\in \CCC$,
  hence $\srbd_\CCC(F_n)=0$. The incidence graph of $T$ contains
  $G_{t+1}$ as a minor, hence $T$'s incidence treewidth is
  $\geq t+1$. Consequently, any $\WWW_t$\hy backdoor of $F_n$ must contain
  at least one variable from each of the $n$ copies of $F_{t+1}$. Thus
  $\sbs_{\WWW_t}(F_n)\geq n$.  It follows that $\sbs_{\WWW_t}$ does
  not dominate $\srbd_\CCC$.

  For the converse direction, consider again the formula $\Chain_d$
  with $\srbd_\CCC(\Chain_d)=d$ (Lemma~\ref{lem:chain}). The incidence
  treewidth of $\Chain_d$ is 2, hence $\sbs_{\WWW_t}(\Chain_t)=0$ for
  $t\geq 2$.  Consequently, $\srbd_\CCC$ does not dominate
  $\sbs_{\WWW_t}$.
\end{proof}


}

\section{Conclusion}

We show that \SAT can be solved in linear-time for formulas of bounded
$\CCC$-backdoor depth whenever $\CCC$ is any of the well-known
Schaefer classes. We achieve this by showing that $\CCC$-backdoor depth can be
FPT-approximated for any class $\CCC=\CCC_{\alpha,s}$. This allows us
to extend the results of \MSV for the class of variable-free formulas
to all Schaefer classes of bounded and notably also unbounded clause
lengths. Our results provide an important milestone towards
generalizing and unifying the various tractability results based on
variants of $\CCC$-backdoor size (see also future work below) to the only recently introduced and
significantly more powerful $\CCC$-backdoor depth. 

There further are natural and potentially significant extensions of
backdoor depth that can benefit from our approach based on separator obstructions. Two of the
probably most promising ones that have already been
successfully employed as extensions of backdoor size are the
so-called \emph{scattered} and \emph{heterogeneous} backdoor
sets~\cite{GaspersMisraOrdyniakSzeiderZivny17,GanianRamanujanSzeider17a};
also refer to \cref{subs:hetero,subs:scattered}. Interestingly,
while those two notions lead to orthogonal tractable classes in the
context of backdoor size, they lead to the same tractable class for backdoor
depth. Therefore, lifting these two extensions to backdoor depth,
would result in a unified and significantly more general
approach. While we are hopeful that our techniques can be adapted to
this setting, one of the main remaining obstacles is that obstructions
of depth $0$ are no longer single (bad) clauses. For instance,
consider the heterogeneous class $\CCC=\Horn\cup \Krom$. Then, the
reason that a CNF formula is not in $\CCC$ can be a pair of
clauses, one in $\Horn\setminus \Krom$ and another one in
$\Krom\setminus \Horn$. Finally, another even more general but also more
challenging tractable class to consider for backdoor depth is the
class of \QHorn{} formulas (see \cref{ssec:qhorn}), which
generalizes the heterogeneous class obtained as the union of all considered
Schaefer classes.

\bibliography{literature}

\ifshort
\clearpage
\appendix

\section*{Appendix}

\section{Proofs from~\cref{sec:srbd}}

\printProofs[ssrbd]

\section{Proofs from~\cref{sec:overview}}

\printProofs[soverview]

\printProofs[sobstructionTrees]

\section{Proofs from~\cref{sec:separatorObstructions}}

\printProofs[sseparatorObstructions]

\section{Proofs from~\cref{sec:find}}

\printProofs[sfind]

\printProofs[scomp]

\fi

\end{document}